\documentclass[letterpaper,11pt]{article}

\usepackage{microtype}

\usepackage[utf8]{inputenc}
\usepackage[T1]{fontenc}
\usepackage[dvipsnames,usenames]{xcolor}
\usepackage[colorlinks=true,urlcolor=Blue,citecolor=Green,linkcolor=BrickRed]{hyperref}
\usepackage{amsmath,amssymb,amsthm,mathabx}
\usepackage{url}
\usepackage{todonotes}
\usepackage[noadjust]{cite}
\usepackage{mathtools}
\usepackage[inline]{enumitem}
\usepackage{algorithm}
\usepackage{algorithmicx}
\usepackage[noend]{algpseudocode}

\usepackage{authblk}
\usepackage{fullpage}

\title{Edit Distance between Unrooted Trees in Cubic Time}
\date{}
\author[1]{Bartłomiej Dudek}
\author[1]{Paweł Gawrychowski}
\affil[1]{Institute of Computer Science, University of Wrocław, Poland}

\newcommand{\myparagraph}{\subparagraph*}

\newcommand{\Oh}{{O}}
\newcommand{\eps}{\varepsilon}
\newcommand{\TT}{T}

\newcommand{\TAB}{\Delta}
\newcommand{\ALLL}[1]{\delta(\cdot,#1)}
\newcommand{\ALLLs}[1]{\delta(*,#1)}
\newcommand{\ALLR}[1]{\delta(#1,\cdot)}
\newcommand{\LH}{P}

\newcommand{\TUS}{\TT_1^{u(*)}}
\newcommand{\mm}{m_{\LH}}

\newcommand{\group}{\textsc{Group}}
\newcommand{\processheavypath}{\textsc{ProcessHeavyPath}}

\newcommand{\computefrom}{\textsc{ComputeFrom}}
\newcommand{\insidegroup}{\textsc{InsideGroup}}

\newcommand{\ldepth}{\mathsf{ldepth}}
\newcommand{\subtree}{\mathsf{subtree}}
\newcommand{\merged}{\mathsf{merged}}

\newcommand{\apex}{\mathsf{apex}}

\newcommand{\bet}{I}
\newcommand{\spec}{\mathsf{special}}
\newcommand{\data}{\mathsf{Data}}

\algnewcommand\algorithmicforeach{\textbf{for each}}
\algdef{S}[FOR]{ForEach}[1]{\algorithmicforeach\ #1\ \algorithmicdo}
\SetLabelAlign{Left}{\makebox[1.5em][l]{#1}}

\newdimen{\algindent}
\algnewcommand\LeftComment[2]{%
\hspace{#1\algindent}$\triangleright$ {#2} \hfill %
}

\newtheorem{theorem}{Theorem}[section]

\newtheorem{lemma}[theorem]{Lemma}

\newtheorem{definition}[theorem]{Definition}

\theoremstyle{definition}

\newtheorem{observation}[theorem]{Observation}

\newcommand{\FIGURE}[4]{
\begin{figure}[#1]
\begin{centering}
\includegraphics[scale={#2}]{{#3}.pdf}
\caption{#4}
\label{fig:#3}
\end{centering}
\end{figure}
}

\begin{document}
 
\maketitle

\begin{abstract}
Edit distance between trees is a natural generalization of the classical edit distance between strings, in which the
allowed elementary operations are contraction, uncontraction and relabeling of an edge.
Demaine et al. [ACM Trans. on Algorithms, 6(1), 2009] showed how to compute the edit distance
between rooted trees on $n$ nodes in $\Oh(n^{3})$ time. However, generalizing their method to
unrooted trees seems quite problematic, and the most efficient known solution remains
to be the previous $\Oh(n^{3}\log n)$ time algorithm by Klein [ESA 1998]. Given the lack
of progress on improving this complexity, it might appear that unrooted trees are simply more
difficult than rooted trees. We show that this is, in fact, not the case, and edit distance between
unrooted trees on $n$ nodes can be computed in $\Oh(n^{3})$ time.
A significantly faster solution is unlikely to exist, as Bringmann et al. [SODA 2018]
proved that the complexity of computing the edit distance between rooted trees cannot be
decreased to $\Oh(n^{3-\eps})$ unless some popular conjecture fails, and the lower bound
easily extends to unrooted trees.
We also show that for two unrooted trees of size $m$ and $n$, where $m\le n$, our algorithm can
be modified to run in $\Oh(nm^2(1+\log\frac nm))$. This, again, matches the complexity achieved
by Demaine et al. for rooted trees, who also showed that this is optimal if we restrict ourselves
to the so-called decomposition algorithms.
\end{abstract}


\section{Introduction}

Computing the edit distance between two strings~\cite{Wagner74} is the most well-known
example of dynamic programming. Thanks to the new fine-grained complexity paradigm,
we known that this simple approach is essentially the best possible~\cite{BackursI15,AbboudHWW16},
so the problem appears to be solved from the theoretical perspective. However, in many
real-life applications we would like to operate on more complicated structures than strings.
As a prime example, while primary structure of RNA can be seen as a string, computational
biology is often interested in comparing also secondary structures. Second structure of RNA
can be modeled as an ordered tree~\cite{ShapiroZ90,HochsmannTGK03}, so we would like
to generalize computing the edit distance between strings to computing the edit
distance between ordered trees.

Tai~\cite{Tai} defined the edit distance between two ordered trees as the minimum total cost of a sequence of elementary operations that transform one tree into the other.
For unrooted trees, which are the focus of this paper, the trees are edge-labeled, and we have three elementary operations: contraction, uncontraction and relabeling of an edge.
We think that the trees are embedded in the plane, i.e., there is a cyclic order on the neighbors of every node that is preserved
by the contraction/uncontraction. See Figure~\ref{fig:contraction}.
The cost of an operation depends on the label(s) of the edge(s): $c_{del}(\tau)$, $c_{ins}(\tau)$, $c_{match}(\tau_1,\tau_2)$, respectively.
We assume that every operation has the same cost as its reverse counterpart: $c_{del}(\tau)=c_{ins}(\tau)$, $c_{match}(\tau_1,\tau_2)=c_{match}(\tau_2,\tau_1)$, and each edge participates in at most one elementary operation.

\FIGURE{b}{0.7}{contraction}
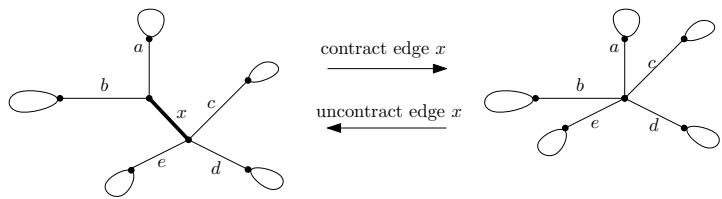
{Contraction and uncontraction of the edge with label $x$ costs $c_{del}(x)=c_{ins}(x)$.}

Computing the edit distance between trees is used as a measure of similarity
in multiple contexts. The most obvious, given that some biological
structures resemble trees, is computational biology~\cite{ShapiroZ90}.
Others include comparing XML data~\cite{BKG03,XML3,Ferragina}, programming
languages~\cite{Hoffmann82}. Others, less obvious, include computer
vision~\cite{BellandoK99,Klein:CV,Klein:ShockGraphs,Klein:ShapeMatching},
character recognition~\cite{Rico-JuanM03}, automatic grading~\cite{AlurDGDV},
and answer extraction~\cite{Yao2013answer}. See also the survey by Bille~\cite{Bille2005}.

Tai~\cite{Tai} introduced the edit distance between rooted node-labeled trees
on $n$ nodes and designed an $\Oh(n^{6})$ algorithm. Zhang and Shasha~\cite{Shasha}
improved the time complexity to $\Oh(n^{4})$ by designing a recursive formula, which
reduces computing the edit distance between two trees to computing the edit distance
between two smaller trees. Then, Klein~\cite{Klein} considered the more general
problem of computing the edit distance between unrooted edge-labeled trees
and further improved the complexity to $\Oh(n^{3}\log n)$ using essentially
the same formula, but applying it more carefully to restrict the number of different
trees that appear in the whole process. This high-level idea of using the recursive
formula can be formalized using the notion of
decomposition strategy algorithms as done by Dulucq and Touzet~\cite{DT05}.
Finally, Demaine et al.~\cite{DMRW} further
improved the complexity for rooted node-labeled trees to $\Oh(n^{3})$.
For trees of different sizes $m$ and $n$, where $m\le n$, their algorithm runs in $\Oh(nm^2(1+\log\frac nm))$ time. 
At a very high level, the gist of their improvement was to apply the heavy path
decomposition to both trees, while in Klein's algorithm only one tree is decomposed.
This requires some care, as switching from being guided by the heavy
path decomposition of the first tree to the second tree cannot be done too
often.

Although Demaine et al.~\cite{DMRW} showed that their algorithm is optimal
among all decomposition strategies, it is not clear that any algorithm must be based on such a strategy. Nevertheless, there
has been no progress on beating the best known $\Oh(n^{3})$ time worst-case
bound for exact tree edit distance. Pawlik and Augsten~\cite{Pawlik15} presented an
experimental comparison of the known algorithms. Aratsu et al.~\cite{AratsuHK10},
Akutsu et al.~\cite{Akutsu10}, and Ivkin~\cite{Ivkin12} designed approximation algorithms.
Only very recently a convincing explanation for the lack of progress on improving
this worst-case complexity has been found by Bringmann et al.~\cite{TED_LowerBound},
who showed that a significant improvement on the cubic time complexity for rooted
node-labeled trees is rather unlikely: an $\Oh(n^{3-\eps}$ algorithm for computing the edit distance
between rooted trees on $n$ nodes implies an $\Oh(n^{3-\eps})$ algorithm for APSP
(assuming alphabet of size $\Theta(n)$) and an $\Oh(n^{k(1-\eps)})$ algorithm for
Max-Weight $k$-Clique (assuming alphabet of sufficiently large but constant size).

Thus, the complexity of computing the edit distance between rooted trees
seems well-understood by now. However, in multiple important applications,
the trees are, in fact, unrooted. For example, Sebastian et al.~\cite{Klein:ShockGraphs}
use unrooted trees to recognize shapes (in a paper with over 700 citations).
Unfortunately, while the almost 20 years old algorithm presented by Klein works for unrooted trees in $\Oh(n^{3}\log n)$
time, it is not clear how to translate Demaine et al.'s improvement to the unrooted
case. In fact, even if one of the trees is a rooted full binary tree and the other is a simple
caterpillar, their approach appears to use $\Oh(n^{4})$ time, and
it is not clear how to modify it. Given the lack of further progress, it might seem that unrooted trees are simply more difficult than rooted trees.

\myparagraph{Our contribution.}
We present a new algorithm for computing the edit distance between unrooted trees which runs in $\Oh(n^3)$ time and $\Oh(n^2)$ space.
For the case of trees of possibly different sizes $n$ and $m$ where $m\le n$, it runs in $\Oh(nm^2(1+\log\frac nm))$ time and $\Oh(nm)$ space.
This matches the complexity of Demaine et al.'s algorithm for the rooted case and improves Klein's algorithm for the unrooted case.
By a simple reduction, unrooted trees are as difficult as rooted trees, so our algorithm is optimal among all decomposition algorithms~\cite{DMRW},
and significantly faster approach is unlikely to exists unless some popular conjecture fails~\cite{TED_LowerBound}.

Our starting point is dynamic programming using the recursive formula of Zhang and Shasha, similarly
as done by Klein and Demaine et al. (in Appendix~\ref{se:demaine} we present a self-contained
description of the latter with a simpler analysis). However, instead of presenting the computation
in a top-down order, we prefer to work bottom-up. This gives us more control and allows us to be more precise
about the details of the implementation. In the simpler $\Oh(n^{3}\log\log n)$ version of the algorithm, we apply the
heavy path decomposition to both trees. As long as the first tree is sufficiently big, we proceed similarly as Klein, that is,
look at its heavy path decomposition.
However, if the first tree is small (roughly speaking) we consider the heavy path decomposition of the second tree and design
a new divide and conquer strategy that is applied on every heavy path separately.

To improve the complexity to $\Oh(n^{3})$, instead of a global parameter we modify the divide and conquer strategy
so that the larger the first tree is the sooner it terminates and switches to another approach. 
A careful analysis of such modification leads to $\Oh(nm^2(1+\log^2 \frac nm))=\Oh(n^3)$ running time.
Finally, we shave one $\log\frac{n}{m}$ by making the divide and conquer sensitive to the sizes of the
subtrees attached to the heavy path instead of its length, that is, making some nodes more important
than the other, reminiscing the so-called telescoping trick~\cite{Blum:Telescoping,Cole:DictionaryMatching}.
All the improvements applied together lead to the overall $\Oh(nm^2(1+\log \frac nm))$ running time, which matches the complexity of the algorithm for rooted trees by Demaine et al. \cite{DMRW}.

\myparagraph{Roadmap.}
In Section~\ref{se:preliminaries} we introduce the notation and the recursive formula that are then used to present Klein's algorithm adapted for the rooted case.
Next, in Section~\ref{se:back_to_unrooted} we return to the unrooted case, introduce new notation and transform both input trees by adding some auxiliary edges.
Then, in Section~\ref{se:unrooted_loglog} we present our new $\Oh(n^3\log\log n)$ algorithm for the unrooted case which already improves the state-of-the-art Klein's algorithm and is essential for understanding our main $\Oh(n^3)$ algorithm described in Section~\ref{se:n^3}.
Both algorithms are described in a bottom-up fashion. In the simpler $\Oh(n^{3}\log\log n)$ version we first assume that one of the trees is a caterpillar and then generalize to arbitrary trees. In the more complicated $\Oh(n^{3})$ algorithm we start with an even more restricted
case of one tree being a caterpillar and the other a rooted full binary tree. When analyzing both algorithms we only bound the total number of
considered subproblems. As explained in Section~\ref{se:implementation}, this can be translated into an implementation with the same running
time. Finally, in Section~\ref{se:lower_bound}, we transfer the known lower bounds to the unrooted case.

\section{Preliminaries}\label{se:preliminaries}

We are given two unrooted trees $\TT_1,\TT_2$ with every edge labeled by an element of $\Sigma$ and a cyclic order on the neighbors of every node.
For every label $\alpha\in\Sigma$, we know the cost $c_{del}(\alpha)=c_{ins}(\alpha)$ of contracting or uncontracting of an edge with label $\alpha$.
For every $\alpha,\beta\in\Sigma$, $c_{match}(\alpha,\beta)=c_{match}(\beta,\alpha)$ is the cost of changing the label of an edge from $\alpha$ to $\beta$.
All costs are non-negative and each edge can participate in at most one operation.
Edit distance between $\TT_1$ and $\TT_2$ is defined as the minimum total cost of a sequence of the above operations transforming $\TT_1$ to $\TT_2$.
Equivalently, it is the minimum cost of transforming both the trees to a common tree using only contracting and relabeling operations, as each operation has the same cost as its undo-counterpart.
Note that for unrooted trees, edit distance is the minimum edit distance over all possible rootings of $\TT_1$ and $\TT_2$, where a rooting is uniquely determined by choice of the root and the leftmost edge from the root.

We first assume, that both trees are of equal size $n=|\TT_1|=|\TT_2|$, but later we will also address the case when one of them is significantly larger than the other.
We start with the case when both trees are rooted, which is essential for the understanding of the unrooted case.
Then, every node has its children ordered left-to-right.
We also assume that both (rooted) trees are binary, as we can add $\Oh(n)$ edges with a fresh label that costs $0$ to contract and $\infty$ to relabel.

\myparagraph{Naming convention.}
We use a similar naming convention as in \cite{DMRW}.
We call main left and right edges of a (rooted) tree respectively the leftmost and rightmost edge from the root.
For a given rooted tree $\TT$ with at least $2$ nodes, let $r_\TT$ denote the right main edge of $\TT$ and $R_\TT$ denote the rooted subtree of $\TT$ that is under (not including)~$r_\TT$.
By $\TT - r_\TT$ we denote a tree obtained from $\TT$ by contracting edge $r_\TT$ and by $\TT - R_\TT$ a tree obtained from $\TT$ by contracting edge $r_\TT$ and all edges from its subtree $R_\TT$.
Thus the tree $\TT$ consists of $R_\TT$, the edge $r_\TT$ and edges $(\TT-R_\TT)$.
$l_\TT$ and $L_\TT$ are defined analogously and $\TT^v$ denotes subtree of $\TT$ rooted at $v$.
See Figure~\ref{fig:naming_and_rec}(a) and (b).

We define a pruned subtree of a tree $\TT$ to be the tree obtained from $\TT$ by a sequence of contractions of the left or right main edge.
Note that every pruned subtree is uniquely represented by the pair of its left and right main edges. 
It also corresponds to an interval on the Euler tour of the tree started in the root when we remove from the inverval each edge that occurs once.
Thus we can completely represent a pruned subtree in $\Oh(1)$ space by storing two edges.
We can preprocess all the $\Oh(n^2)$ pruned subtrees $\TT'$ of a tree $\TT$ to be able to obtain trees $R_{\TT'},L_{\TT'},\TT'-l_{\TT'},\TT'-r_{\TT'}$ and edges $r_{\TT'},l_{\TT'}$ in $\Oh(1)$ time.

\myparagraph{Dynamic programming.}
Zhang and Shasha \cite{Shasha} introduced the following recursive formula for computing the edit distance between two rooted trees:
\begin{lemma}\label{le:dp}
Let $\delta(F,G)$ be the edit distance between two pruned subtrees $F$ and $G$ of respectively $\TT_1$ and $\TT_2$. Then:
\begin{itemize}
 \item $\delta(\emptyset,\emptyset)=0$
 \item $\delta(F,G)= \min
 \begin{cases}
  \delta(F-r_F,G)+c_{del}(r_F) & \text{ if } F\ne\emptyset\\
  \delta(F,G-r_G)+c_{del}(r_G) & \text{ if } G\ne\emptyset\\
  \delta(R_F,R_G)+\delta(F-R_F,G-R_G)+c_{match}(r_F,r_G) & \text{ if } F,G\ne\emptyset\\
 \end{cases}$
\end{itemize}
The above recurrence also holds if we contract or match the left main edge.
\end{lemma}

It contracts the right main edge in one of the two trees or matches the right main edges of the two trees. 
In the latter case, we get two independent subproblems $(R_F,R_G)$ and $(F-R_F,G-R_G)$ that must be transformed to equal trees.
See Figure~\ref{fig:naming_and_rec}(c) for an illustration of this case.

\FIGURE{b}{0.74}{naming_and_rec}
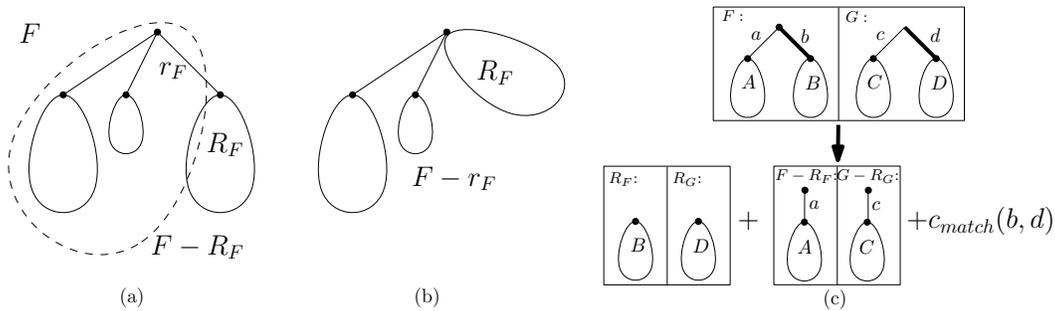
{(a) Tree $F$ with both $r_F$ and $R_G$ contracted. (b) $F$ with its right main edge contracted. (c) When both right main edges are not contracted we obtain two independent problems.}

To estimate time complexity of the algorithm, we only count different pairs $(F,G)$ for which $\delta(F,G)$ is computed. Each such value
is computed at most once and stored. Note that $F$ is always a pruned subtree of $\TT_1$,
while $G$ is a pruned subtree of $\TT_2$, thus there are $\Oh(n^4)$ possible pairs $(F,G)$.
In the worst case, all such pairs might be considered. The formula from Lemma~\ref{le:dp} can
be evaluated in constant time, and any previously computed value can be retrieved in constant
time from a four-dimensional table.

The above algorithm always contracts or relabels the right main edge. A more deliberate choice
of direction (whether to choose the left or right main edge) will lead to a different behavior
of the algorithm which in turn might result in a smaller total number of considered pairs $(F,G)$.
Such a family of algorithms is called decomposition algorithms.
When analyzing the time complexity of such an algorithm, we assume that any already computed 
$\delta(F,G)$ can be retrieved in constant time. If our goal is to compute significantly fewer
than $\Oh(n^4)$ subproblems, we cannot afford to allocate the four-dimensional table anymore.
An obvious solution is to store the already computed values in a hash table, but this requires
randomization. In Section~\ref{se:limited_space} we explain how to carefully arrange the order of
the computation and store the partial results as to obtain deterministic algorithms with the same 
running time.

While the formula from Lemma~\ref{le:dp} suggests a top-down strategy, we phrase the
algorithms in a bottom-up perspective, which allows us to present the details of the
computation more precisely.
The aim of all the algorithms is to compute $\delta(\TT_1,\TT_2)$ knowing only $\delta(\emptyset,\cdot)$ and $\delta(\cdot,\emptyset)$, as the costs of contraction of an arbitrary pruned
subtree are precomputed.

\myparagraph{Klein's $\Oh(n^3\log n)$ algorithm.}
Klein's algorithm \cite{Klein} uses heavy path decomposition \cite{SleatorTarjan} of~$\TT_1$.
The root is called light and every node calls its child with the largest subtree (and the leftmost in case of ties) heavy and all other children light.
An edge is heavy if it leads to the heavy child.

While applying the dynamic formula from Lemma~\ref{le:dp}, Klein's algorithm uses a strategy that we call ``avoiding the heavy child'' in $\TT_1$.
It chooses the direction (either left or right) in such a way that the edge leading to the heavy child of the root is contracted or relabeled as late as possible.
Observe that contracting the main edge not leading to the heavy child of the root of a pruned subtree $\TT$, does not change the heavy child of the root of $\TT$, as its subtree is still the largest.
Note that Klein's strategy does not depend on the considered pruned subtree of $\TT_2$.

Even though Klein uses top-down view to describe his algorithm, we find it more convenient to implement the computations in bottom-up order.
Therefore the algorithm processes heavy paths of $\TT_1$ in the bottom-up order as shown in Algorithm~\ref{alg:klein}.
Consider a heavy path $H$ with nodes $v_1,v_2,\ldots,v_{|H|}$ where $v_1$ is the closest node to the root and $v_{|H|}$ is a leaf. By $\ALLR{\TT_1^v}$ we denote a table of $\Oh(n^2)$ distances between tree $\TT_1^v$ and all pruned subtrees of $\TT_2$.
The algorithm considers all nodes on $H$ also bottom-up.
It starts from $\ALLR{\TT_1^{v_{|H|}}}=\ALLR{\emptyset}$, which is precomputed, and then iteratively computes $\ALLR{\TT_1^{v_i}}$ from $\ALLR{\TT_1^{v_{i+1}}}$ for decreasing values of $i$.
We denote such a step by $\computefrom$ subroutine.
Note that in every step the strategy avoiding the heavy child always chooses the same direction (recall that the tree is binary) and visits altogether at most $\Oh(n)$ pruned subtrees of $\TT_1$.
Also when actually implementing the $\computefrom$ step we proceed bottom-up.
That is, suppose we have already computed $\ALLR{\TT_1^{v_{i+1}}}$ and that $v_{i+1}$ is the left child of $v_i$.
Then the strategy avoiding the heavy child says R that is chooses first the right main edge to consider.
We compute $\ALLR{\TT_1^{v_i}}$ as follows.
First we consider the tree $\TT_1^{v_{i+1}}\cup \{\{v_i,v_{i+1}\}\}$ (we call this uncontracting the heavy edge), next $\TT_1^{v_{i+1}}\cup \{\{v_i,v_{i+1}\},\{v_i,w\}\}$ if exists a light child $w$ of $v_i$ and then uncontract the subsequent edges of $\TT_1^w$.
This guarantees that while computing $\delta(F,G)$ the subtrees $F-r_F$ and $F-R_F$ have been already processed.
Pruned subtrees of $\TT_2$ are also considered in the order of increasing sizes.
Clearly, as argued for Zhang and Shasha's algorithm, the algorithm visits $\Oh(n^2)$ pruned subtrees of $\TT_2$, so we need to bound the number of pruned subtrees of $\TT_1$.

\begin{algorithm}[t]
\begin{algorithmic}[1]
  \ForEach{heavy path $H$ in $\TT_1$ in the bottom-up order}
  \State let $v_1,v_2,\ldots,v_{|H|} = H$ 
  \For{$i=|H|-1,\ldots,0$}
    \Statex \LeftComment{2}{avoiding the heavy child:}
    \State $\computefrom(\ALLR{\TT_1^{v_i}},\ALLR{\TT_1^{v_{i+1}}})$
  \EndFor
  \EndFor
\end{algorithmic}
\caption{Klein's algorithm.}
\label{alg:klein}
\end{algorithm}

\begin{observation}\label{obs:big_cuts}
 Consider an arbitrary tree $\TT$. Suppose that strategy avoiding the heavy child in $\TT$ says $R$ for a pruned subtree~$F$.
 Then $F-R_F$ is also obtained by a sequence of contractions of the main edge according to the strategy.
\end{observation}

The observation implies that in order to count the relevant intervals of $\TT_1$ we can only consider the trees obtained by contraction of the main edge according to the strategy, and trees of the form $L_F$ and $R_F$.
Note that the only trees of the form $L_F$ or $R_F$ that are not obtained in this way are rooted at a light node so will be counted separately for another heavy path. 

We denote $\apex(F)$ to be the top node of the heavy path containing the lowest common ancestor of all endpoints of edges of $F$.
In other words, $\apex(F)$ is the lowest light ancestor of all edges of~$F$.
Now grouping all the visited pruned subtrees by their $\apex$-es we bound their total number:

\begin{observation}
 For an arbitrary tree $\TT$, there is $\sum_{v:\mathrm{\ light\ node\ in\ }\TT} |\TT^v|$ pruned subtrees of $\TT$ visited while applying strategy avoiding the heavy child of $\TT$.
\end{observation}

Let light-depth $\ldepth(u)$ of a node $u$ be the number of light nodes that are ancestors of $u$ (node is also an ancestor of itself).
Because $\ldepth(u)\le~\log(n)~+~1$, we obtain:

\begin{equation}\label{eq:light_nodes}
 \sum_{v:\text{ light node in }\TT_1} |\TT_1^v| = \sum_{v:\text{ node in }\TT_1} \ldepth(v) \in\Oh(n\log n)
\end{equation}

\noindent
Recalling that there are $\Oh(n^2)$ relevant intervals of $\TT_2$ we conclude that Klein's algorithm visits $\Oh(n^3\log n)$ subproblems.
As we assume the constant time memoization, it runs in $\Oh(n^3\log n)$ time.

\section{Back to Unrooted Case}\label{se:back_to_unrooted}

Recall that edit distance between two unrooted trees $\TT_1$ and $\TT_2$ is the minimum edit distance between $\TT_1$ and $\TT_2$ over all possible rootings of them, where rooting is determined by the root of the tree and its the left main edge.
As Klein \cite{Klein} mentioned, it is enough to choose an arbitrary rooting in one of the trees and try all possible rootings of the other to find an optimal setting.
Observe, that we can treat the Euler tour of $\TT_2$ as a cyclic string and represent every pruned subtree of $\TT_2$ as an interval of it, for all possible rootings of $\TT_2$.
Thus Klein's algorithm works in $\Oh(n^3\log n)$ time also for the edit distance between unrooted trees.
Before we present our faster algorithm for this case, we need to introduce some new definitions.
Recall, that even in the unrooted case, we first arbitrarily root both trees and the initial rooting remains unchanged throughout the algorithm.

\myparagraph{Darts.}
We replace every edge $e$ with two darts corresponding to two ways of traversing the edge, either down $e^\downarrow$ or up the tree $e^\uparrow$ (with respect to the fixed rooting).
Subtree of a dart $\subtree((u,v))$ is defined as the subtree rooted at node $v$, when $u$ is its parent.
Note that $e^\uparrow$ and $e^\downarrow$ belong neither to $\subtree(e^\uparrow)$ nor to $\subtree(e^\downarrow)$.
Every pruned subtree of (unrooted) tree is uniquely represented by its left and right main edges or a dart (if there is one edge from the root).
See Figure~\ref{fig:repr}.

\FIGURE{t}{0.81}{repr}
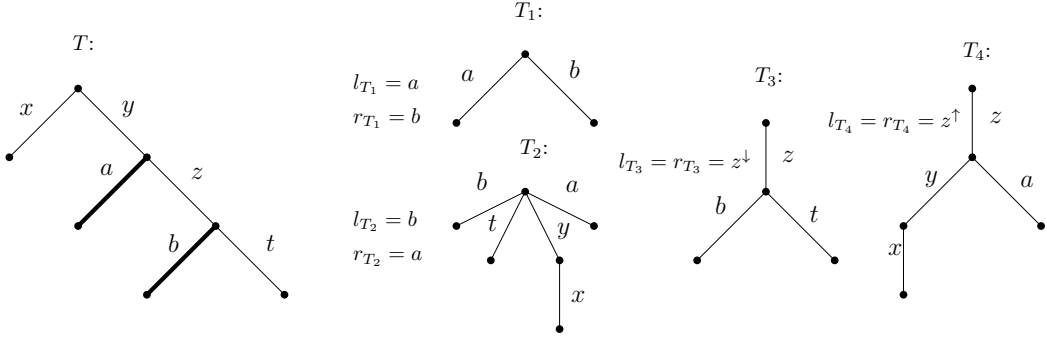
{Every pruned subtree is uniquely represented by its left and right main edges or a dart.}

\myparagraph{Auxiliary edges for rootings.}
We observed that every rooting of $\TT_2$ corresponds to a subrange of a cyclic Euler tour $E_{\TT_2}$, but later it will be convenient to represent every rooting as a subtree of a dart.
For this purpose, we add new edges labeled with a fresh label $\# \notin\Sigma$ which will be used only to denote a rooting.
Setting $c_{del}(\#)=0$ and $c_{match}(\#,\cdot)=\infty$ we force that these edges are only contracted.
For every node $v$ we add new edges alternating with the original ones.
Thus in total, there are $2(n-1)$ edges added.
Using these new edges we can compute edit distance between the unrooted trees from the values of $\delta(\TT_1,\subtree(d))$ for all darts $d$ in $\TT_2$. Thus, our aim is to fill the table $\TAB$ where $\TAB[u,d]:=\delta(\TT_1^u,\subtree(d))$.

\myparagraph{Auxiliary edges to bound the degrees.}
As the last step, again we add $\Oh(n)$ edges with appropriate costs as to ensure that the degree of every node is at most $3$. Observe that the cost of the optimal solution for the modified trees is the same as for the initial ones and having a sequence of operations for the modified trees, we can easily obtain an optimal sequence for the original instance of the problem.

\section{\texorpdfstring{$\Oh(n^3\log\log n)$}{O(n3loglogn)} Algorithm for Unrooted Case}\label{se:unrooted_loglog}

After initial modifications both trees are binary and the algorithm needs to fill the table $\TAB$ where $\TAB[u,d]:=\delta(\TT_1^u,\subtree(d))$ for all nodes $u\in\TT_1$ and darts $d\in\TT_2$.
We first run Demaine et al.'s algorithm operating on labels on edges instead of nodes which computes $\delta(\TT_1^u,\TT_2^v)$ for all nodes $u\in\TT_1$ and $v\in\TT_2$ in $\Oh(n^3)$ time and stores them in $\TAB[u,e_v^\downarrow]$ where $e_v^\downarrow$ is the dart to $v$ from its parent.
Now we need to fill the remaining fields $\TAB[u,e^\uparrow]$ for all darts $e^\uparrow$ up the tree $\TT_2$.

This is the main difficulty in the unrooted case, in which we need to handle many big subtrees which are significantly different from each other.
Our approach is to successively reduce different subproblems to smaller ones, in a way that there are fewer subproblems to consider in the next step.
We use divide and conquer paradigm, in which there is more and more sharing after every step.

In the beginning, we call each node of $\TT_1$ and $\TT_2$ light or heavy as in the Klein's algorithm and all the time the notion is with respect to the initial rootings.
Similarly, the notion of traversing an edge up or down the tree is always with respect to the rooting.
Recall that we denote $\apex(\TT)$ as the top node on the heavy path containing the lowest common ancestor of all edges of $\TT$.
We first fix a global value $b$, which will be determined exactly later.
On a high level, from the top-down perspective, the algorithm uses the following strategy to compute $\delta(F,G)$:
if $|\TT_1^{\apex(F)}|>n/b$, then avoid the heavy child in $F$, and otherwise apply a new strategy based only on $G$ and $\TT_2$.

Considering it bottom-up, the algorithm first fills values of $\TAB[u,e^\uparrow]$ for all nodes $u$ such that $|\TT_1^{\apex(u)}|\le n/b$  and all darts up $\TT_2$.
For the remaining fields of $\TAB$, it uses strategy avoiding the heavy child in $\TT_1$.
As in the Klein's algorithm, in this phase, the algorithm needs to process heavy paths of $\TT_1$ in the bottom-up order.
Note that for each light node $v$ such that $|\TT_1^v|>n/b$ holds $\ldepth(v)<\log b +1$.
Thus there are $\Oh(n^3\log b)$ subproblems visited in total in this phase.

For the other phase note that there are $\Oh(n^2/b)$ relevant subtrees in $\TT_1$, and now we need to carefully design and analyze the new strategy for $\TT_2$.
It will be easier to think, that in this phase the algorithm needs to compute $\TAB[u,e^\uparrow]$ for all darts $e^\uparrow$ up $\TT_2$ and all nodes $u\in\TT_1$ such that $|\TT_1^u|\le n/b$, call them interesting.
Clearly, all subproblems in which there is a switch to the strategy based on $\TT_2$ are of this form.

As now the strategy will be more complex than before, we first describe it for the case when $\TT_2$ is a caterpillar: a heavy path with possibly single nodes connected to it.
This example is already difficult in the unrooted case and will require divide and conquer approach to handle all the possible rootings of $\TT_2$ at once.
Next, we will slightly modify the approach to handle arbitrary trees $\TT_2$.

\subsection{Caterpillar \texorpdfstring{$\TT_2$}{T2}}

Now we consider the case when $\TT_2$ is a heavy path $H$ with possibly single nodes connected to it.
Let $h_i$ denote (heavy) edges on $H$, $r_2=h_0$ be the edge denoting the initial rooting of $\TT_2$ and (if exists) $l_i$ be the light edge connected to the $i$-th node on $H$.
See Figure~\ref{fig:shp_and_up_from_heavy}(a) for an example. 

\FIGURE{t}{0.8}{shp_and_up_from_heavy}
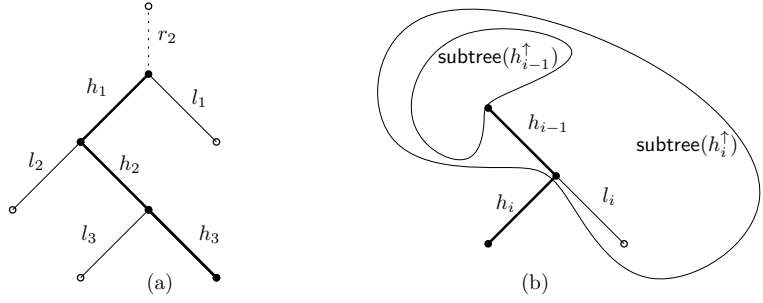
{(a) A heavy path $H$ with edge $r_2$ (dotted) denoting the rooting of $\TT_2$.
(b) To compute $\ALLLs{\subtree(h_i^\uparrow)}$ we use $\ALLLs{\subtree(h_{i-1}^\uparrow)}$, first uncontract the edge $h_{i-1}$ and then $l_i$ (if exists).}

In the first step we compute values of $\ALLLs{\subtree(h_i^\uparrow)}$ for all heavy edges $h_i$, where $*$ denotes all pruned subtrees of $\TT_1$ of size at most $n/b$.
The strategy is to avoid the parent, that is to contract the edge leading to the parent as late as possible.
See Figure~\ref{fig:shp_and_up_from_heavy}(b).

More precisely, in the beginning, we already know $\ALLLs{\subtree(h_0^\uparrow)}$, because it is the cost of contraction of the whole pruned subtree of $\TT_1$ (which is precomputed), as $h_0=r_2$ and $\subtree(h_0^\uparrow)=~\emptyset$.
Then, having values of $\ALLLs{\subtree(h_{i-1}^\uparrow)}$ we compute $\ALLLs{\subtree(h_i^\uparrow)}$ by uncontracting first $h_{i-1}$ and then $l_i$ if it exists.
It is an extension of the $\computefrom$ subroutine, but now we do not have subtrees $\TT^x$ and $\TT^y$, where $x$ is the parent of $y$, but two edges $h_i$ and $h_{i-1}$ with a common endpoint.
Note that in this step all uncontractions are from the same direction.

There are $\Oh(n)$ pruned subtrees of $\TT_2$ obtained by uncontractions of a main edge according to the strategy, starting from the empty subtree.
Now we need to show that the algorithm did not consider any other pruned subtree of $\TT_2$.
Suppose it uncontracted the left main edge.
Then $G-L_G\in\{\emptyset,G-l_G\}$, depending on whether $l_G$ was the heavy edge leading to the parent or not.
Also $L_G\in\{\emptyset,G-l_G\}$, so in both cases, all the obtained pruned subtrees are among the $\Oh(n)$ described above.
Finally, as there are $\Oh(n^2/b)$ pruned subtrees of $\TT_1$, in total we computed and stored the edit distance of $\Oh(n^3/b)$ subproblems.
Now, using the computed values we fill $\TAB[u,h_i^\uparrow]$ for all interesting nodes $u\in\TT_1$ and heavy edges $h_i\in\TT_2$.
Thus, later on, we do not have to consider the pruned subtrees of the form $\delta(L_F,L_G)$ or $\delta(R_F,R_G)$ as their values are already stored in $\TAB$, because either one of them is empty or they are of the form $\delta(\TT_1^v,\subtree(dh))$ for an interesting node $u\in\TT_1$ and a dart $dh$ from a heavy edge in $\TT_2$.
We only have not computed values $\TAB[u,l^\uparrow]$ for darts from light edges up the tree, but in this phase of the algorithm, they never appear in $\delta(L_F,L_G)$ or $\delta(R_F,R_G)$ subproblem.
However, we need to compute these values because they correspond to some rootings of $\TT_2$, so we will consider them in the following paragraph.

\myparagraph{Darts from light nodes up the tree.}
From now on, our algorithm processes heavy paths of $\TT_2$ one-by-one.
In particular, in this subsection, we process the only heavy path $H$ of $\TT_2$.
Thus, unless explicitly stated otherwise all the notion is relative to the current heavy path~$H$.
First, we define $\merged^R(A,B)$ as the pruned subtree obtained by contraction of edges between the $A$-th and $B$-th node on $H$ or to the right of $H$:

\begin{definition}
  Let $H$ be a heavy path and $A$ and $B$ ($A\le B$) denote indices of two nodes on~$H$.
  Then $\merged^R(A,B)$ is a tree with the left main edge $h_{A-1}$ and the right main edge~$h_B$.
  $\merged^L(A,B)$ is a tree with the left main edge $h_B$ and the right main edge $h_{A-1}$.
\end{definition}
\noindent
See Figure~\ref{fig:merging}.
Note that $\subtree(l_A^\uparrow)$ is either $\merged^R(A,A)$ or $\merged^L(A,A)$, depending on which side of $H$ is $l_A$.

\FIGURE{t}{0.8}{merging}
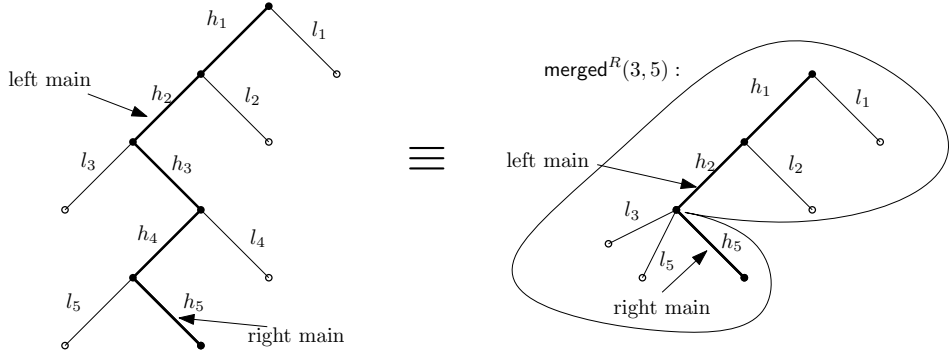
{Pruned subtree $\merged^R(3,5)$ has the left main edge $h_2$ and right $h_5$.}

As explained earlier, in the beginning the algorithm computes $\ALLLs{\subtree(h_i^\uparrow)}$ for all heavy edges on $H$.
Additionally, it calculates $\ALLLs{\merged^L(1,|H|)}$ and $\ALLLs{\merged^R(1,|H|)}$ from $\ALLLs{\subtree(h_0^\uparrow)}$ by repeatedly uncontracting respectively the right and left main edge.
See Algorithm~\ref{alg:process_heavy_path} for the summary of the whole preprocessing.
Then it calls a recursive procedure $\group(1,|H|,\data(1,|H|))$.
The final goal of this call is to fill $\TAB[u,l_i^\uparrow]$ for all light edges $l_i$ connected to the heavy path~$H$.

\begin{algorithm}[bht]
\begin{algorithmic}[1]
  \Function{$\processheavypath$}{$\ALLLs{\subtree(h_0^\uparrow)}$}
  \For{$i=1..|H|$}
    \Statex \LeftComment{2}{avoiding the parent:}
    \State $\computefrom(\ALLLs{\subtree(h_i^\uparrow)},\ALLLs{\subtree(h_{i-1}^\uparrow}))$
    \State fill $\TAB[u,h_i^\uparrow]$ for all interesting nodes $u$
  \EndFor
  \Statex
  \Statex \LeftComment{1}{repeatedly uncontracting the left main edge:}
  \State $\computefrom(\ALLLs{\merged^R(1,|H|)},\ALLLs{\subtree(h_0^\uparrow)})$
  \Statex \LeftComment{1}{repeatedly uncontracting the right main edge:}
  \State $\computefrom(\ALLLs{\merged^L(1,|H|)},\ALLLs{\subtree(h_0^\uparrow)})$
  \Statex
  \State $\group(1,|H|,\data(1,|H|))$
 \EndFunction
\end{algorithmic}
\caption{Computes input tables needed for processing a heavy path~$H$}
\label{alg:process_heavy_path}
\end{algorithm}

$\group(A,B,\data(A,B))$ is a procedure which considers an interval $[A,B]$ of indices on $H$ given tables of values $\ALLLs{\subtree(h_{A-1}^\uparrow)},\ALLLs{\subtree(h_{B}^\downarrow)},\ALLLs{\merged^L(A,B)}$ and $\ALLLs{\merged^R(A,B)}$, which we denote as $\data(A,B)$. 
Intuitively, $\data(A,B)$ contains information about subtrees ``outside'' the considered interval $[A,B]$ which are relevant during intermediate computations.
Then, the procedure calls itself recursively for shorter intervals until it holds that
$A=B$ when $\ALLLs{\merged^L(A,A)}$ or $\ALLLs{\merged^R(A,A)}$ contains the fields of $\TAB[u,l_A^\uparrow]$ for all interesting nodes $u$ and then the recurrence stops.

In more detail, for an interval $[A,B]$, the procedure computes $\data(A,M)$ and $\data(M+1,B)$ for $M = \lfloor \frac{A+B}{2} \rfloor$ and calls itself recursively for the smaller intervals.
Note that for $\data(A,M)$ it needs to compute tables $\ALLLs{G}$ for trees $G=\merged^R(A,M),\merged^L(A,M)$ or $\subtree(h_M^\downarrow)$ and can reuse
table $\ALLLs{\subtree(h_{A-1}^\uparrow)}$ which is a part of $\data(A,B)$.
Similarly for interval $[M+1,B]$.
See Algorithm~\ref{alg:grouping}.

\begin{algorithm}[h]
\begin{algorithmic}[1]
  \Function{$\group$}{$A,B,\data(A,B)$}
  \If{$A=B$}
    \If{there is a light edge $l_A$ connected to $H$}
      \State fill $\TAB[u,l_A^\uparrow]$ for interesting nodes $u\in\TT_1$ \label{line:end_of_rec}
    \EndIf
    \State \Return
  \EndIf
  \State $M:=\lfloor(A+B)/2)\rfloor$
 
  \For{$i=(B-1)..M$} \label{line:grouping_easy1}
    \Statex \LeftComment{2}{avoiding the heavy child:}
    \State $\computefrom(\ALLLs{\subtree(h_i^\downarrow)},\ALLLs{\subtree(h_{i+1}^\downarrow)})$  \label{line:grouping_easy1_compute_from}
  \EndFor
  
  \Statex \LeftComment{1}{repeatedly uncontracting the right main edge:}
  \State $\computefrom(\ALLLs{\merged^R(A,M)},\{\ALLLs{\merged^R(A,B)};\ALLLs{\subtree(h_{A-1}^\uparrow)}\})$ \label{line:grouping_hard1}
  
  \Statex \LeftComment{1}{repeatedly uncontracting the left main edge:}
  \State $\computefrom(\ALLLs{\merged^L(A,M)},\{\ALLLs{\merged^L(A,B)};\ALLLs{\subtree(h_{A-1}^\uparrow)}\})$ \label{line:grouping_hard11}
  \Statex
  \State  $\group(A,M,\data(A,M))$ 
  \State symmetric computations for interval $[M+1,B]$ \label{line:grouping_hard22}
  \State  $\group(M+1,B,\data(M+1,B))$
 \EndFunction
\end{algorithmic}
\caption{Fills $\TAB[u,l_i^\uparrow]$ for light edges $l_i$ connected to the heavy path~$H$ with $i\in[A,B]$.}
\label{alg:grouping}
\end{algorithm}

To analyze the complexity of the $\group$ procedure, first note that in every step of the loop in line~\ref{line:grouping_easy1}, it considers a constant number of pruned subtrees from $\TT_2$, so in total there are  $\Oh(B-M)$ of them.
After this loop, we have $\ALLLs{\subtree(h_{M}^\downarrow)}$ computed.

The call of $\computefrom$ in line~\ref{line:grouping_hard1} needs more input than the call in line~\ref{line:grouping_easy1_compute_from}, even though the strategy is always uncontracting the right main edge. Note that if the dynamic program only tried contracting the right main edge, it would be possible to compute $\ALLLs{\merged^R(A,M)}$ only from $\ALLLs{\merged^R(A,B)}$.
However, it is not the case when the algorithm also matches edges.
The first case when $r_G$ is a light edge ($r_G=l_X$ for some value of $X$) is not problematic, because then $R_G=\emptyset$ and $G-R_G=G-r_G$, so this pruned subtree is already visited.
Although, if $r_G$ is a heavy edge then $R_G=\subtree(r_G^\downarrow)$ and $G-R_G$ is a pruned subtree, which has not been considered yet.
Observe that in this situation the pruned subtree can be obtained from $\subtree(h_{A-1}^\uparrow)$ 
by a sequence of $\Oh(B-A)$ contractions of the right main edge, so we need it as a separate input to the $\computefrom$ subroutine.
A similar reasoning applies to the edges to the left of $H$ in line~\ref{line:grouping_hard11} and to the computations for interval $[M+1,B]$.


To sum up, one call of $\group(A,B)$ (not including recursive calls) visits $\Oh(B-A)$ pruned subtrees of $\TT_2$.
As we start from an interval of length $|H|$ and in every recursive call its length is roughly halved, the procedure considers in total $\Oh(|H|\log |H|)=\Oh(n\log n)$ pruned subtrees of $\TT_2$.

\subsection{Arbitrary Tree \texorpdfstring{$\TT_2$}{T2}}\label{se:loglog_arbitrary_t2}

Now we describe, how to modify the above algorithm to process not only a caterpillar, but an arbitrary tree $\TT_2$.
In this case, there can be non-empty subtrees connected to the main heavy path.

Note that for an arbitrary heavy path $H$ inside $\TT_2$, the $\processheavypath$ procedure only needs to know $\ALLLs{\subtree(h_0^\uparrow)}$ to be able to compute all the remaining input parameters in $\data(1,|H|)$, because $\ALLLs{\subtree(h_{|H|}^\downarrow)}=\ALLLs{\emptyset}$ is precomputed.
In the beginning, the algorithm calls $\processheavypath_{H^0}(\ALLLs{\emptyset})$, where $H^0$ is the heavy path of $\TT_2$ containing the root of $\TT_2$.
The only place we need to change inside the $\group$ procedure to handle arbitrary trees $\TT_2$ is to not only fill $\TAB[u,l_A^\uparrow]$ in line~\ref{line:end_of_rec} of Algorithm~\ref{alg:grouping}, but also recursively call $\processheavypath_{H'}(\ALLLs{\subtree(l_A^\uparrow)})$ where $H'$ is the heavy path connected to the $A$-th node of the considered heavy path.
As we pointed earlier, $\subtree(l_A^\uparrow)$ is either $\merged^R(A,A)$ or $\merged^L(A,A)$, depending on which side of $H$ is $l_A$.
Now observe, that each subsequent pruned subtree that appears in the recursive formula is already visited and processed:

\begin{observation}\label{le:nothing_missed}
In the modified $\group$ procedure, during the call of $\computefrom$ subroutine in line~\ref{line:grouping_hard1} of Algorithm~\ref{alg:grouping}, all the intermediate pruned subtrees of $\TT_2$ are obtained by a sequence of uncontractions of the right main edge from the root either from $\merged^R(A,B)$ or $\subtree(h_{A-1}^\uparrow)$.
A similar property holds for the other three calls of $\computefrom$ in lines \ref{line:grouping_hard11} and \ref{line:grouping_hard22}.
\end{observation}
%
%

What changes in the analysis of the procedure is that now there are not $\Oh(|H|\log |H|)$ pruned subtrees of $\TT_2$ but $\Oh(|\TT_2^v| \log |H|)=\Oh(|\TT_2^v| \log n)$, where $v$ is the top node of $H$.
In other words, the heavy path $H$ itself might be short, but there might be big subtrees connected to it.
However, every subtree connected to $H$ is completely contracted (edge-by-edge) a constant number of times on every level of recursion of $\group$ procedure and thus the bound.

Recall that top node of every heavy path is light, so using equation \eqref{eq:light_nodes} we bound the overall number of subtrees of $\TT_2$ considered during this part of the algorithm:

$$\sum_{v:\text{ top node of a heavy path in }\TT_2} |\TT_2^v|\cdot\log n = \sum_{v:\text{ light node in }\TT_2} |\TT_2^v|\cdot \log n \in\Oh(n\log^2 n)$$

\subsection{Final Analysis}

To conclude, the above algorithm computes $\TAB[u,e^\uparrow]$ for all nodes $u\in\TT_1$ such that $|\TT_1^u|\le n/b$ and all darts up the tree $\TT_2$ by considering $\Oh(n\log^2 n)$ pruned subtrees of $\TT_2$ and $\Oh(n^2/b)$ of $\TT_1$.
At the beginning of Section~\ref{se:unrooted_loglog} we described the second phase of the algorithm, which avoids the heavy child in $\TT_1$ and fills the remaining fields of $\TAB$ considering $\Oh(n\log b)$ pruned subtrees of $\TT_1$ and $\Oh(n^2)$ of $\TT_2$.
Thus, during the two phases, the whole algorithm visits $\Oh(n^3\frac{\log^2 n}{b}+n^3\log b)$ subproblems.
Setting $b=\log^2 n$ we obtain the overall complexity $\Oh(n^3\log\log n)$.

%
%
%

\section{Optimal \texorpdfstring{$\Oh(n^3)$}{O(n3)} Algorithm for Unrooted Case}\label{se:n^3}

We start with transforming both trees as in the $\Oh(n^3\log\log n)$ algorithm, that is we add auxiliary edges for rootings, root them arbitrarily and finally make them binary.
Let $\TT_1$ and $\TT_2$ denote the transformed trees.
From now on we assume that $|\TT_1|\le |\TT_2|$ (otherwise we swap the trees) and let $m=|\TT_1|,n=|\TT_2|$.
In this section, we present an algorithm that computes the edit distance between unrooted trees in $\Oh(nm^2(1+\log\frac nm))=\Oh(n^3)$ time.
Again we assume constant-time access to values of $\delta(F,G)$ for all the already considered subproblems.
It can be obtained using i.e. hashing, but in Section~\ref{se:implementation} we will focus on implementation details and show how to fit all the computations in $\Oh(nm)$ space without randomization.

As in the $\Oh(n^3\log\log n)$ approach, the algorithm aims to fill the table $\TAB[u,d]:=\delta(\TT_1^u,\subtree(d))$ from which it computes the answer to the original problem.
In the beginning it runs Demaine et al.'s algorithm \cite{DMRW} which computes $\delta(\TT_1^u,\TT_2^v)$ for all pairs of nodes $u\in\TT_1$ and $v\in\TT_2$ in $\Oh(nm^2(1+\log\frac nm))$ time.
Now it remains to compute $\TAB[v,e^\uparrow]$ for all nodes $v\in\TT_1$ and all darts up the tree $\TT_2$.

The algorithm first decomposes both trees into heavy paths.
Then it processes heavy paths in $\TT_1$ bottom-up.
To avoid confusion, a heavy path of $\TT_1$ we denote by $P$ and of $\TT_2$ by $H$.
For every heavy path $\LH$ in $\TT_1$ the algorithm fills $\TAB[v,e^\uparrow]$ for all nodes $v\in \LH$ and darts up the tree $\TT_2$ using modified $\processheavypath$ procedure.
Now there is no global parameter $b$, but instead of that, the algorithm uses $\mm$, the size of the subtree rooted at the top node of $\LH$.
See Algorithm~\ref{alg:turns}.

\begin{algorithm}[H]
\begin{algorithmic}[1]
  \ForEach{$\LH:$ heavy path in $\TT_1$ bottom-up}
      \State $u:=\text{ top node of }\LH$
      \State \textbf{global} $\mm:=|\TT_1^u|$
      \State process heavy paths of $\TT_2$ and compute $\TAB[\TT_1^v,e^\uparrow]$ for all $v\in \LH$ and $e^\uparrow\in\TT_2$
  \EndFor
\end{algorithmic}
\caption{Computes the answer in phases.}
\label{alg:turns}
\end{algorithm}

We call all the computations for one heavy path of $\TT_1$ a phase.
In the following, we describe in detail a single phase. 
As the presentation is involved, we break it into pieces and gradually handle more and more difficult cases.

\myparagraph{Roadmap.}
In the beginning, similarly as in the $\Oh(n^3\log\log n)$ algorithm, we first focus on the case when $\TT_2$ is a heavy path with single connected nodes, as in  Figure~\ref{fig:shp_and_up_from_heavy}(a).
It already highlights the difficulties that we will encounter while obtaining $\Oh(nm^2(1+\log\frac nm))$ complexity.
We also assume that $\TT_1$ is a full binary tree, which simplifies the analysis, because there are roughly $2^k$ heavy paths with size $m/2^k$.

Next, we relax the assumption on $\TT_1$ and consider an arbitrary tree $\TT_1$.
The change does not affect the algorithm at all, but now we know less about sizes of the heavy paths, which changes the analysis.
Using a technical lemma we show that, even in this case, the algorithm also runs in $\Oh(nm^2(1+\log\frac nm))$ time.

Then, we adapt the algorithm to handle arbitrary trees $\TT_2$, as in the $\Oh(n^3\log\log n)$ approach.
The direct generalization runs in $\Oh(nm^2(1+\log^2\frac nm ))$ time, which is already $\Oh(n^3)$, but still slower than Demaine et al.'s algorithm, which runs in $\Oh(nm^2(1+\log\frac nm ))$ time.
The next step is to change the way of dividing the interval in the divide and conquer approach by taking into account sizes of subtrees connected to the considered heavy path have different sizes.
This improvement finally leads to $\Oh(nm^2(1+\log\frac nm))$ running time, which we believe to be optimal.

In Section~\ref{se:implementation} we describe how to implement this algorithm in $\Oh(nm)$ space.

\subsection{Full Binary Tree and Caterpillar}\label{se:fbt_comb}

We first describe the phase for a single heavy path $\LH$ of $\TT_1$ for the case when $\TT_1$ is a full binary tree and $\TT_2$ is a single heavy path.
Recall that $u$ is the top node of $\LH$, we defined $\mm=|\TT_1^u|$ and let $H$ be the main single heavy path of $\TT_2$.

In the beginning, the algorithm runs similarly as in the $\processheavypath$ subroutine of $\Oh(n^3\log\log n)$ approach.
Recall that we computed tables $\ALLLs{G}$ for some pruned subtrees $G$ of $\TT_2$ where $*$ denotes all subtrees of $\TT_1$ of size at most $n/b$.
Now we will also compute similar tables, but for all pruned subtrees of $\TT_1^u$ where $u$ is the top node of the considered heavy path $P$ of $\TT_1$.
Then there are $\Oh(\mm^2)$ such trees and we will denote them as $\cdot$ in $\ALLL{G}$, a table of $\Oh(\mm^2)$ values.

The algorithm first fills $\TAB[u,h_i^\uparrow]$.
Then it runs a modified recursive procedure $\group$ which stops the recursion when the length of the considered interval $[A,B]$ is smaller than $\mm$, differently than the $A=B$ condition of Algorithm~\ref{alg:grouping}.
See Algorithm~\ref{alg:grouping_to_m}.

\begin{algorithm}[H]
\begin{algorithmic}[1]
  \Function{$\group$}{$A,B,\data(A,B)$}
  \If{$B-A<\mm$}
    \State $\insidegroup(A,B,\data(A,B))$
    \State \Return
  \EndIf
  \State $M:=\lfloor(A+B)/2)\rfloor$
  \State compute intermediate values and call itself recursively as in Algorithm~\ref{alg:grouping}
 \EndFunction
\end{algorithmic}
\caption{A slight change in the divide and conquer approach.}
\label{alg:grouping_to_m}
\end{algorithm}

In the base case of the $\group$ recursion, when $B-A < \mm$, the algorithm has already computed $\ALLL{\subtree(h_{A-1}^\uparrow)}$, $\ALLL{\subtree(h_{B}^\downarrow)}$, $\ALLL{\merged^L(A,B)}$ and $\ALLL{\merged^R(A,B)}$ but this is not sufficient to fill all the missing fields of~$\TAB$, as $B>A$.
We denote all the subsequent computations in the base case as $\insidegroup$ procedure.
In order to describe them in detail, we first introduce some auxiliary notation. 

\myparagraph{Auxiliary notation.}
While considering an interval $[A,B]$, let $I(A,B)$ be the set of edges in $\TT_2$ that are ``between'' $h_{A-1}$ and $h_B$, formally: $I(A,B)=~\subtree(h_{A-1}^\downarrow) \setminus (\subtree(h_B^\downarrow)\cup\{h_B\})$.
We will write simply $I$ if the interval $[A,B]$ is clear from the context.
See Figure~\ref{fig:inside_edges} for an example.
Let $D=\{h_{A-1},h_B\}$ be the set of the boundary edges and $T_{2D}$ be the set of four trees with both main edges in $D$.
Later on, we will be interested only in the pruned subtrees with both main edges in $D\cup I$.
While considering the subtrees in the $\insidegroup$ procedure, we never need to access subproblems with subtrees $(\subtree(h_B^\uparrow)+\{h_B\})$ or $(\subtree(h_{A-1}^\downarrow)+\{h_{A-1}\})$.
Therefore, among all $6$ trees with both main edges in $D$ we consider only the following four
$T_{2D}~=~\{G_1,G_2,(\subtree(h_{A-1}^\uparrow)+~\{h_{A-1}\}),(\subtree(h_B^\downarrow)+\{h_B\})\}$ where trees $G_1$ and $G_2$ satisfy $l_{G_1}=h_{A-1},r_{G_1}=h_B$ and $l_{G_2}=h_{B},r_{G_2}=h_{A-1}$ respectively.

\FIGURE{H}{1}{inside_edges}
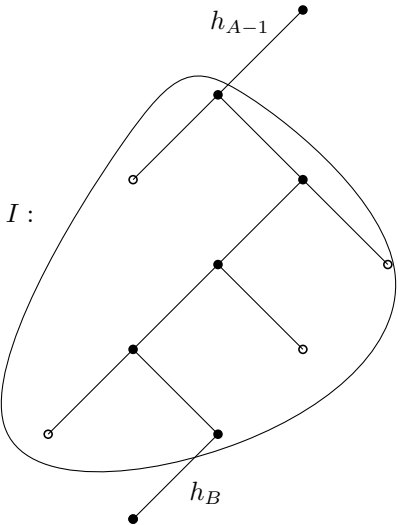
{$I(A,B)$ is the set of edges ``between'' $h_{A-1}$ and $h_B$.}

To describe a set of pruned subtrees with particular main edges we only write the condition they satisfy.
If one edge is not specified we assume it belongs to $I$, for instance $[l_G=h_B]$ we read $\{G: l_G=h_B \wedge r_G\in I\}$.
Finally, $\TT_1^{u(x)}$ denotes the tree $\TT_1^u$ after $x$ contractions according to the strategy avoiding the heavy child in $\TT_1$ and $\TUS$ denotes all possible pruned subtrees of this form: $\TUS=\{\TT^{u(i)}:i\le |\TT^u|\}$.

\begin{algorithm}[h]
\begin{algorithmic}[1]

\Function{$\insidegroup$}{$A,B,\data(A,B)$}
  \State $\computefrom(\ALLL{\subtree(h_{B}^\uparrow)+\{h_{B}\}}, \ALLL{\subtree(h_{B}^\uparrow)})$\label{line:inside1a}
  \State $\computefrom(\ALLL{\subtree(h_{A-1}^\uparrow)+\{h_{A-1}\}}, \ALLL{\subtree(h_{A-1}^\uparrow)})$\label{line:inside1b}
  \Statex

  \Statex \LeftComment{1}{repeatedly uncontracting the left main edge:}
  
  \State $\computefrom(\ALLL{[r_G=h_{A-1}]},\{\ALLL{\merged^L(A,B)},\ALLL{\subtree(h_{A-1}^\uparrow)}\})$\label{line:inside2_beg}
  
  \State $\computefrom(\ALLL{[r_G=h_B]},\{\ALLL{\merged^R(A,B)},\ALLL{\subtree(h_B^\downarrow)}\})$ 
  \Statex \LeftComment{1}{repeatedly uncontracting the right main edge:}
  \State $\computefrom(\ALLL{[l_G=h_{A-1}]},\{\ALLL{\merged^R(A,B)},\ALLL{\subtree(h_{A-1}^\uparrow)}\})$ 
  \State $\computefrom(\ALLL{[l_G=h_B]},\{\ALLL{\merged^L(A,B)},\ALLL{\subtree(h_B^\downarrow)}\})$\label{line:inside2_end} 
   
  \Statex  
  \For{$i=(|\TT_1^u|-1)..0$} \label{line:loop_hc}
    \Statex \LeftComment{2}{avoiding the heavy child in $\TT_1$:}
    \State $\computefrom(\delta(\TT_1^{u(i)},[l_G,r_G\in I]),\delta(\TT_1^{u(i+1)},[l_G,r_G\in I]))$ \label{line:avoid_hc} 
    \If{exists $v:\TT_1^{u(i)}=\TT_1^v$}
    \State fill $\TAB[v,e^\uparrow]$ for all edges $e\in I$
    \EndIf
  \EndFor
  \EndFunction
\end{algorithmic}
\caption{Fills $\TAB[v,e^\uparrow]$ for all edges $e\in I$ and nodes $v$ on $\LH$.}
\label{alg:inside_a_group}
\end{algorithm}

\myparagraph{Procedure.} Algorithm~\ref{alg:inside_a_group} runs in three steps considering subtrees of $\TT_2$ with respectively both, only one and no main edges in $D$.
\begin{enumerate}
 \item Both main edges in $D$. Recall that we consider only the four trees from $\TT_{2D}$. Then it is enough to compute $\ALLL{\subtree(h_{A-1}^\uparrow)+\{h_{A-1}\}}$ and $\ALLL{\subtree(h_{B}^\downarrow)+\{h_B\}}$, because the other two tables are part of $\data(A,B)$. See lines~\ref{line:inside1a}-\ref{line:inside1b}.
 \item Exactly one main edge in $D$. We compute every table separately, repeatedly uncontracting edges from the same direction. For this purpose we need to pass two tables to $\computefrom$ subroutine.
 See details in lines \ref{line:inside2_beg}-\ref{line:inside2_end}.
 \item Both main edges in $I$. Now we take into account the considered pruned subtree of $\TT_1$ and use strategy avoiding its heavy child. Our aim is to fill $\TAB[v,e^\uparrow]$ for nodes~$v$ on the considered heavy path $\LH$ of $\TT_1$ and all edges $e\in I$. For this purpose we compute $\delta(\TUS,[l_G,r_G\in I])$ where $u$ is the top node of $\LH$. See lines~\ref{line:loop_hc}-\ref{line:avoid_hc}. 
\end{enumerate}

In the last step it is crucial that when the dynamic formula from Lemma~\ref{le:dp} uses a tree with at least one main edge in $D$, then the value of the considered subproblem has been already computed, stored and can be returned in a constant time. This property is summarized in the following observation.

\begin{observation}\label{obs:edges_in_d}
 For every tree $G$ with both main edges in $I$ it holds that $G-l_G,G-L_G$ and $L_G$ have both main edges in $D \cup I$.
 Moreover, if $G-l_G,G-L_G$ or $L_G$ has both main edges in $D$, then it belongs to $TT_{2D}$.
 A similar property holds for contracting the right main edge.
\end{observation}


\myparagraph{Analysis.}
From Observation~\ref{obs:edges_in_d} we have that all the computations in line~\ref{line:avoid_hc} consider only the trees with both main edges in $D\cup I$ and no other kind of tree can appear.
So there are $\Oh(\mm^2)$ pruned subtrees of $\TT_2$ considered, as $|I|\le 2\mm$ (recall that $B-A<\mm$).
Next, while avoiding the heavy child of $\TT_1^u$ we only consider pruned subtrees of $\TUS$, so there are $\Oh(\mm)$ of them.
Thus, during all the computations in the loop in line~\ref{line:loop_hc} there are $\Oh(\mm^3)$ considered subproblems.
Similarly, in all the earlier computations of the $\insidegroup$ procedure, there are $\Oh(\mm^2)$ pruned subtrees of $\TT_1^u$ and $\Oh(\mm)$ of $\TT_2$, so in total there $\Oh(\mm^3)$ subproblems considered.
Notice that for every group size of its corresponding set $I$ is at least $\mm/2$.
The sets are disjoint, so there are at most $2n/\mm$ groups on the heavy path in total.
Thus, there are $\Oh(n\mm^2)$ pruned subtrees considered in all calls of the $\insidegroup$ procedure.

Now we bound the complexity of the whole $\processheavypath$ procedure, similarly as in the analysis of the $\Oh(n^3\log\log n)$ algorithm.
For one heavy path $P$ from $\TT_1$ there are $\Oh(1+\log \frac {n}{\mm})$ recursive calls of $\group$, because the length of the interval is halved until it gets smaller than $\mm$.
Again, every edge of $\TT_2$ contributes to $\Oh(1)$ pruned subtrees on every level of recursion, so there are $\Oh(n(1+\log \frac {n}{\mm}))$ subtrees of $\TT_2$.
We consider all the $\Oh(\mm^2)$ pruned subtrees of $\TT_1^u$, so all the computations during recursive calls visit $\Oh(n\mm^2(1+\log\frac{n}{\mm}))$ subproblems.
Adding $\Oh(n\mm^2)$ pruned subtrees from the calls of the $\insidegroup$ procedure we conclude that in total, during the whole phase for one heavy path $\LH$ of $\TT_1$, the algorithm considers $\Oh(n\mm^2(1+\log\frac{n}{\mm}))$ subproblems.
Now we need to sum this over all heavy paths $\LH$ in $\TT_1$.
As $\TT_1$ is a full binary tree of size $m$, we can divide the heavy paths into groups where paths $P$ from $i$-th group satisfy $\frac{m}{2^{i+1}}<\mm\le\frac{m}{2^{i}}$.
Then we write:

\begin{align*}
\hspace{-0.5cm}\sum_{\LH\in\TT_1}n\mm^2\left(1+\log \frac{n}{\mm}\right) &\le \sum_{i=0}^{\log m} 2^in\left(\frac {m}{2^i} \right)^2\left(1+\log \frac{n}{m/2^{i+1}}\right) = nm^2\sum_{i=0}^{\log m}  \frac {1}{2^i} \left(2+i+\log \frac{n}{m}\right) \\
&= nm^2\left(\sum_{i=0}^{\log m}  \frac {2+i}{2^i} +\log \frac{n}{m} \sum_{i=0}^{\log m} \frac{1}{2^i}\right) \in \Oh\left(nm^2\left(1+\log \frac nm \right)\right)
\end{align*}
\noindent
To conclude, the algorithm for darts up the tree $\TT_2$ visits in total $\Oh(nm^2(1+\log\frac nm))$ subproblems.
Recall that Demaine et al.'s algorithm and the strategy for all darts up $\TT_2$ from heavy nodes visit the same number of subproblems.
To sum up, the whole algorithm for full binary tree $\TT_1$ and a single heavy path $\TT_2$ runs in $\Oh(nm^2(1+\log\frac nm))$ time.

\subsection{Arbitrary Tree and Caterpillar}

For the case of an arbitrary tree $\TT_1$, the algorithm is the same as above, but now we need a different analysis of the overall running time.
For this purpose, we first analyze properties of the function $f(x)=~x^2(1~+~\ln\frac nx)$ which appears in the complexity of various parts of the algorithm.

\begin{lemma}\label{le:property}
   Let $f(x)=x^2(1+\ln\frac nx)$.
   If $x,y$ satisfy $1\le x\le y<(n-1)/2$, then:
   \begin{enumerate}[label=(\roman*),align=Left]
    \item $f(t)$ is non-decreasing in the range $[1,n/2)$,\label{bullet:1}
    \item $f(1)+f(x) \le f(x+1)$,\label{bullet:2}
    \item if $x>1$ then: $f(x)+f(y)\le f(x-1)+f(y+1)$.\label{bullet:4}
   \end{enumerate}
  \end{lemma}
  \begin{proof}

  We prove \ref{bullet:1} directly by computing derivative of $f$:
  $$\frac{\partial f(t)}{\partial t}=\frac{\partial (t^2(1+\ln\frac nt))}{\partial t}= t\left(1+2\ln\frac nt\right) \ge 1+ 2\ln \frac{n}{n/2} \ge 0$$
  
  \noindent  
  Similarly, \ref{bullet:2} follows from the definition and inequalities: $\ln(x)\ge 1-1/x$ for $x>0$ and $n\ge x+1$:
  \begin{align*}
  f(x+1)-f(x)-f(1)&=(x+1)^2\left(1+\ln\frac{n}{x+1}\right)-x^2\left(1+\ln\frac nx\right)-1-\ln n\\
  &=x^2\ln\frac {x}{x+1} +2x\left(1+\ln\frac{n}{x+1}\right) +1+\ln\frac{n}{x+1} -1-\ln n\\
  &=x^2\ln\frac {x}{x+1} +2x\left(1+\ln\frac{n}{x+1}\right) +\ln\frac{1}{x+1}\\
  &\ge x^2\left(1-\frac{x+1}{x}\right)+2x + (1-(x+1))=0\\
  \end{align*}
  
  \noindent
  To prove \ref{bullet:4} we first show it for $x=y$, that is: if $x>1$ then holds $2f(x)\le f(x-1)+f(x+1)$.
  For this purpose we also need that $n>2x$:
  \begin{align*}
f(x+1)+f(x-1)-2f(x)&=x^2\ln\frac{x^2}{x^2-1}+2x\ln\frac{x-1}{x+1}+2+\ln\frac{n^2}{x^2-1}\\
  &\ge x^2\ln\frac{x^2}{x^2-1}+2x\ln\frac{x-1}{x+1}+2+\ln\frac{(2x)^2}{x^2-1}\\
  &\ge x^2\left(1-\frac{x^2-1}{x^2}\right)+2x\left(1-\frac{x+1}{x-1}\right)+2+\ln 4\\
  &=1 - \frac{4x}{x-1} +2 + \ln 4 \\
  &= \frac{-4}{x-1} + \ln 4 -1 \ge 0 \quad \text{for } x\ge 12
  \end{align*}

  \noindent
  For $x<12$ we calculate the exact value of the expression in the second line which is non-negative.
  As we proved that \ref{bullet:4} holds for $x=y$, now it is enough to show that:
  \[\frac{\partial}{\partial y}\left(f(y+1)+f(x-1)-f(x)-f(y)\right)\ge~0. \]
  This can be done as follows:
  \begin{align*}
  \frac{\partial}{\partial y}\big(f(y+1)+f(x-1)-f(x)-f(y)\big) &= (2y+2)\ln\frac{n}{y+1}-2y\ln\frac{n}{y}+1\\
  &= 2y\ln\frac{y}{y+1}+2\ln\frac{n}{y+1} + 1 \\
  &\ge 2y\left(1-\frac{y+1}{y}\right)+2\ln\frac{n}{y+1} + 1\\
  &=2\ln\frac{n}{y+1}-1 \\
  &\ge2\ln2 -1 \ge 0 \qedhere
  \end{align*}
  \end{proof}

Recall that we need to bound the sum $n\sum_{\LH: \text{ heavy path in }\TT_1} \left(\mm^2\left(1+\log \frac{n}{\mm}\right)\right)$, but now we have no assumptions on the tree $\TT_1$.
For this we will use the following lemma:
  
\begin{lemma}\label{le:sum}
 Let $m$ be size of a tree $\TT$ and $n$ be an arbitrary number such that $n\ge m$.
 Then:
 $$\sum_{\LH:\mathrm{\ heavy\ path\ in\ }\TT} \mm^2\left(1+\log\frac n{\mm}\right) = \Oh\left(m^2\left(1+\log\frac nm \right)\right)$$
\end{lemma}
\begin{proof}
  We start with changing the logarithm to natural to simplify the following calculations, hiding the constant factor under $\Oh$.
  Let $t(m):=\sum_{\LH:\text{ heavy path in }\TT} \mm^2\left(1+\ln\frac n{\mm}\right)$ be the above sum for a tree $\TT$ of size $m$.
  Denoting by $m_i$ the total size of the $i$-th subtree hanging off the heavy path containing the root of $\TT$, we obtain the following bound: 
  \begin{equation}\label{eq:rec_sum}
  t(m)\le m^2\left(1+\ln\frac nm\right)+\sum_{i} t(m_i)\quad \text{ where } m_i\le \frac{m-1}{2} \text{ and } \sum_i m_i \le m-1
  \end{equation}  
  where the sum is over all subtrees connected to the heavy path from the root of $\TT$.
 
  Now we use Lemma~\ref{le:property} to bound from above the sum $\sum_if(m_i)$, in which $m_i\le (m-1)/2$ and $\sum_i m_i \le m-1$ (from \eqref{eq:rec_sum}).
  As long as there are three non-zero $m_i$s, two of them are less than $(m-1)/2$ and we choose distinct indices $i,j$ such that $1\le m_i \le m_j < (m-1)/2$.
  Then, depending whether $m_i$ equals $1$ or not, we apply \ref{bullet:2} or \ref{bullet:4} to decrease $m_i$ and increase $m_j$ by one, not changing sum of the values.
  Hence, in the end, there are at most two non-zero $m_i$'s, and they are less than or equal to $m/2$.
  Next we apply \ref{bullet:1} and finally obtain that: $\sum_if(m_i)\le 2f(m/2)$.
  
  Using the above bound, equation~\eqref{eq:rec_sum} and applying the master theorem we obtain that $t(m)\in\Oh(m^2(1+\ln\frac nm))$.
  \end{proof}

  The lemma finishes the analysis of the algorithm for an arbitrary tree and a single heavy path, which runs in $\Oh(nm^2(1+\log \frac nm))$ time.

\subsection{Both Trees Arbitrary}

Now we consider the case when both trees are arbitrary.
We generalize the algorithm for a single heavy path to arbitrary trees, as in the $\Oh(n^3\log\log n)$ approach.
As earlier, $H^0$ is the heavy path containing the root of $\TT_2$ and in the beginning we call $\processheavypath_{H^0}$, but we need to change the $\group$ subroutine again.
Suppose we consider a heavy path $H$.
Then, even if the considered interval $[A,B]$ is short, still there can be big subtrees attached to the path with significantly more than $\mm$ edges in total.
As earlier, we denote by $I(A,B)$ the set of edges ``between'' $h_{A-1}$ and $h_B$.
See Figure~\ref{fig:inside_edges}.
Now we break the recursion if there is less than $\mm$ edges in $I(A,B)$.
Note that now even if $A=B$ it may happen that we do not call the $\insidegroup$ procedure, because there are at least $\mm$ edges in the subtree attached to the $A$-th node.
Then we denote by $H'$ the heavy path attached to the $A$-th node of $H$ and call $\processheavypath_{H'}$ with appropriate input parameters.
See Algorithm~\ref{alg:grouping_to_m_arbitrary}.

\begin{algorithm}[H]
\begin{algorithmic}[1]
  \Function{$\group$}{$A,B,\data(A,B)$}
  \If{$|\bet(A,B)|<\mm$}
    \State $\insidegroup(A,B,\data(A,B))$ \label{line:inside_group}
    \State \Return
  \EndIf
  \If{$A=B$}
    \State let $H'$ be the heavy path connected to the $A$-th node on $H$
    \State $\processheavypath_{H'}(\ALLL{\subtree(l_A^\uparrow)})$ \label{line:process_h_prim}
    \State \Return
  \EndIf
  \State $M:=\lfloor(A+B)/2)\rfloor$
  \State compute intermediate values and call itself recursively as in Algorithm~\ref{alg:grouping}
 \EndFunction
\end{algorithmic}
\caption{Considering the total number of edges between the $A$-th and $B$-th node on $H$.}
\label{alg:grouping_to_m_arbitrary}
\end{algorithm}

Observe that $\processheavypath$ is called only for heavy paths $H'$ such that $n_{H'} \ge \mm$, because otherwise the line~\ref{line:inside_group} is executed instead of line~\ref{line:process_h_prim}.
In line~\ref{line:process_h_prim} we have already computed $\ALLL{\subtree(l_A^\uparrow)}$, because $\subtree(l_A^\uparrow)$ is either $\merged^L(A,A)$ or $\merged^R(A,A)$, depending on which side of $H$ is $l_A$, so it is a part of $\data(A,A)$.
We need to notice, that during computations in line~\ref{line:inside_group} we do not have computed values $\TAB[u,e^\uparrow]$ for edges $e$ on heavy paths connected to $H$.
However, all subproblems of this form have both main edges in $\bet(A,B)$, so they are considered by the $\insidegroup$ subroutine and the missing fields of $\TAB$ are filled then.

Notice that the changes in the procedure for a single heavy path $H$ of $\TT_2$ do not affect the complexity, which is still $\Oh(n_{H}\mm^2(1+\log\frac{n_{H}}{\mm}))$ where $n_{H}$ is the size of the subtree of $\TT_2$ rooted at the top node of $H$.
Thus all the computations for a single heavy path $\LH$ run in time:
$$\sum_{H: \text{ heavy path in }\TT_2 \text{ s.t. } n_{H}\ge \mm} n_{H}\mm^2\left(1+\log\frac {n_{H}}{\mm}\right) \in \Oh\left(n\mm^2\left(1+\log^2\frac n\mm\right)\right)$$

\noindent
because $\sum_{H}n_{H} = \Oh(n(1+\log\frac n{\mm}))$ as the algorithm considers only heavy paths $H$ of $\TT_2$ such that $n_{H}\ge \mm$.
Now using Lemma~\ref{le:sum} we obtain, that the overall complexity of the algorithm for the general case of both trees arbitrary is:

$$ \sum_{\LH: \text{ heavy path in }\TT_1}  n\mm^2\left(1+\log^2 \frac{n_{H}}{\mm} \right) \in \Oh\left(nm^2\left(1+\log^2\frac {n}{m}\right)\right)$$

\noindent 
which is $ \Oh(n^3)$, as desired.

To further reduce the complexity and obtain $\Oh(nm^2(1+\log\frac nm))$ time, we need to modify the $\group$ procedure slightly.
Our approach reminisces the telescoping trick from \cite{Blum:Telescoping} and \cite{Cole:DictionaryMatching} in which some nodes are less important than the others.
Intuitively, considering an interval $[A,B]$ on a heavy path~$H$, we would like to divide it in such a way that the big subtrees connected to $H$ are contracted smaller number of times.
For this purpose, we need to look at the tree from a different perspective.

\myparagraph{Big nodes.}
Let big nodes be the nodes of $\TT_2$ with subtree containing at least $\mm$ edges.
Then the big nodes form a top part of $\TT_2$ and there are some small nodes connected to them.
See Figure~\ref{fig:big_nodes}.
We also define that a big light node is called special if it has no big light descendant.

\FIGURE{h}{0.9}{big_nodes}
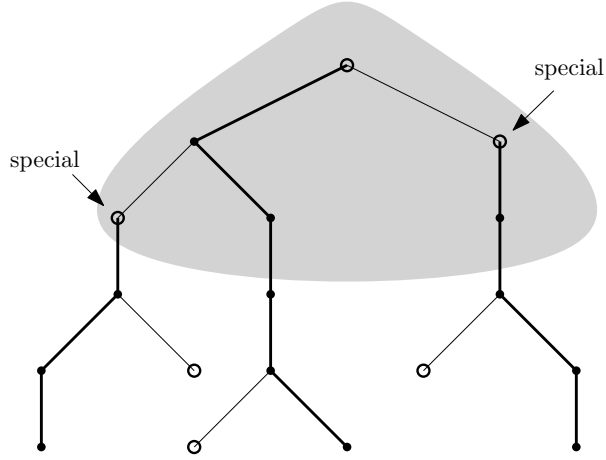
{Big nodes constitute the gray upper part of the tree and a big light node is called special if it has no big light descendant. Heavy edges are thick.}

Now, instead of counting the edges in $\bet(A,B)$ we count the special nodes inside subtrees connected to the $A,A+1,\ldots,B$-th node on $H$ and we denote the value as $\spec_{H}(A,B)$.
Only when this value is $0$, we again focus on $\bet(A,B)$.
See Figure~\ref{fig:big_heavy_paths} for an example of how these values are computed.

\FIGURE{h}{1}{big_heavy_paths}
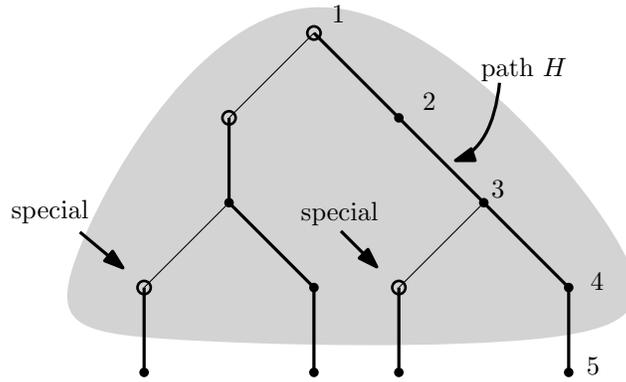
{There are two special nodes in the tree and holds: $\spec_{H}(1,3)=~2$, $\spec_{H}(3,5)=1$ and $\spec_{H}(2,2)=0$.}

Using the notion of special nodes we can describe the algorithm in detail.
If for the considered interval $[A,B]$ it holds that $\spec(A,B)=0$, then the pivot is chosen as earlier: $M:=\lfloor(A+B)/2)\rfloor$.
Otherwise we choose $M$ to be the smallest index such that $2\cdot\spec(A,M)\ge\spec(A,B)$.
See Algorithm~\ref{alg:grouping_trick}.
Notice that now we exclude the $M$-th node from the subsequent recursive calls and run $\group(A,M-1)$ and $\group(M+1,B)$.
This is because we would like the value of $\spec(A,B)$ to decrease by a factor of 2 in subsequent
recursive calls. By the choice of $M$, $2\cdot\spec(M+1,B)\le\spec(A,B)$ but we cannot be
sure that $2\cdot\spec(A,M)\le\spec(A,B)$. However, $2\cdot\spec(A,M-1)\le\spec(A,B)$ surely holds,
so we can recurse on $[A,M-1]$ and $[M+1,B]$ and consider the $M$-th node separately.
To make the pseudo-code more concise we process the $M$-th node by recursively calling $\group(M,M,\data(M,M))$.
This call will subsequently reach line~\ref{line:final_1} or~\ref{line:final_2} and terminate.

\begin{algorithm}[H]
\begin{algorithmic}[1]
  \Function{$\group$}{$A,B,\data(A,B)$}
  \If{$A>B$}
  \State \Return
  \EndIf
  \If{$\spec(A,B)>0$}
    \If{$A=B$}
      \State $\processheavypath_{H'}(\ALLL{\subtree(l_A^\uparrow)})$ \label{line:final_1}
      \State \Return
    \Else
      \State $M:=\min \{k: 2\cdot\spec(A,k)>\spec(A,B)\}$
    \EndIf
  \Else
    \If{$B-A<\mm$}
      \State $\insidegroup(A,B,\data(A,B))$ \label{line:final_2}
      \State \Return
    \Else
      \State $M:=\lfloor(A+B)/2)\rfloor$
    \EndIf
  \EndIf
  \State compute arguments for subsequent recursive calls
  \State $\group(A,M-1,\data(A,M-1))$
  \State $\group(M,M,\data(M,M))$
  \State $\group(M+1,B,\data(M+1,B))$
 \EndFunction
\end{algorithmic}
\caption{The final version of the recursive $\group$ procedure.}
\label{alg:grouping_trick}
\end{algorithm}

Recall that all the subroutines $\processheavypath,\insidegroup$ and the intermediate computations $\computefrom$ run in $\Oh(s\cdot\mm^2)$ time (not including subsequent calls), where $s=|I(A,B)|$ is the number of edges currently considered.
So now we need to sum the number of edges considered in all recursive calls together.

\myparagraph{Final analysis.}
Consider an edge $e\in\TT_2$.
Note that all calls $\insidegroup$ consider disjoint sets of edges, so $e$ can contribute to at most one of them. 
Now we count triples $(H,A,B)$ such that $e\in I_{H}(A,B)$, that is the edge $e$ is considered in the recursive call $\group_{H}(A,B)$.
Let $X_e$ be the set of such triples.
Observe, that there are at most $\Oh(1+\log \frac n{\mm})$ heavy paths $H$ on the path from $e$ to the root of $\TT_2$ because we consider only heavy paths of size at least $\mm$.
See Figure~\ref{fig:group_in_group}.
Thus we have an upper bound on the number of triples with $A=B$ in~$X_e$.

\FIGURE{h}{1}{group_in_group}
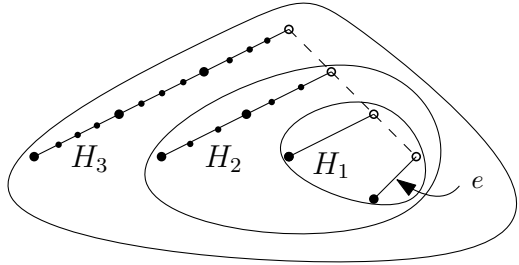
{An edge $e$ is considered by all big heavy paths ``above'' it: $H_1,H_2,H_3$ and there are $\Oh(1+\log \frac n{\mm}))$ of them.}

Now we aggregately count triples $(H,A,B)\in X_e$ such that $A\ne B$.
First, observe that there is at most one heavy path $H'$, such that $e\in I_{H'}(A,B)$ and $\spec_{H'}(A,B)=0$ because all the heavy paths ``above'' have at least one big heavy path connected to them that contains~$e$.
In this case, there are at most $\Oh(1+\log \frac {n}{\mm})$ recursive calls, as every time length of the interval is roughly halved and cannot become smaller than $\mm$.

Second, as we pointed earlier, every time $\spec_{H}(A,B)\ne 0$, in every subsequent recursive call this value is at least halved.
As all the subtrees of special nodes are disjoint and contain at least $\mm$ edges, there are at most $\frac{n}{\mm}$ of them.
Thus, there are also $\Oh(1+\log \frac {n}{\mm})$ recursive calls with $\spec_{H}(A,B)\ne 0$.

Finally, at the base of the recursion, we call $\insidegroup$ procedure.
Recall that if the size of the considered set of edges is $x$ then the complexity of the procedure is $\Oh(\mm^2\cdot x+\mm\cdot x^2)$.
As all the calls are applied to disjoint subsets of edges, each of them consists of at most $\mm$ edges and in total there are $\Oh(n)$ of them, the total complexity of all these calls is $\Oh(n \mm^2)$.

To conclude, during the phase for a heavy path $\LH$ of $\TT_1$, every edge of $\TT_2$ is considered in at most $\Oh(1+\log \frac {n}{\mm})$ recursive calls. 
All the intermediate computations inside a recursive call require $\Oh(\mm^2)$ time per edge of $\TT_2$, hence the whole phase for $\LH$ runs in $\Oh(n\mm^2(1+\log \frac {n}{\mm}))$ time.
Finally we use Lemma~\ref{le:sum} and conclude that the whole algorithm computing edit distance between unrooted trees runs in $\Oh(nm^2(1+\log \frac {n}{m}))$ time.

\section{Implementation Details}\label{se:implementation}

Currently, the above algorithm runs in $\Oh(nm^2(1+\log\frac nm))$ time and space if we use hashing to store already computed subproblems.
In this section, we show how to deterministically implement it in $\Oh(nm)$ space in the same time.
Recall that $\TT_2$ is not smaller than $\TT_1$.
There will be three difficulties to face.

First, now we cannot preprocess all the pruned subtrees of $\TT_2$ because $\Oh(n^2)$ space is already too much for us.
Thus, given a pruned subtree $G$ of $\TT_2$, we need to be able to retrieve in a constant time subtrees $G-l_G,G-L_G,L_G,\ldots$ and the value of $\delta(\emptyset,G)$ (the cost of contraction of all the edges from $G$).
For that purpose, we will use the classic algorithm for Lowest Common Ancestor \cite{Bender2000} that runs in linear space and answers queries for the lowest common ancestor ($LCA$) of two nodes in constant time.

Second, we need to show how to implement the $\computefrom$ subroutine in $\Oh(nm)$ space, which is an order of magnitude less than the number of subproblems considered in the subroutine: $\Oh(nm^2)$.
This step will be done similarly as in Demaine et al.'s algorithm, even though now we consider the unrooted case.

Finally, we need to take into account the depth of the recursion of $\group$ procedure and count how much data is kept on the stack.
We will show, that on every level of recursion there is $\Oh(m^2)$ data stored.
As we proved, there are $\Oh(1+\log \frac nm)$ levels of recursion, so using inequality $\log x \le x$ we get that the total memory kept on the stack is $\Oh(m^2(1+\log\frac nm))=\Oh(nm)$.

\subsection{Preprocessing}\label{se:preproc}

Recall that every pruned subtree tree $G$ is represented by its left and right main edges ($l_G$ and $r_G$).
If they overlap, then the tree is of the form $\subtree(d)$ for some dart of $\TT_2$.
There are only $\Oh(n)$ trees of this form, so we can preprocess them all, that is for every pruned subtree $G=\subtree(d)$ we store $\delta(\emptyset,G)$ and the pruned subtrees: $G-l_G, G-L_G,L_G,\ldots$.
Now we focus on one rooting of $\TT_2$ but do not have to decompose it into heavy paths.
We first run the preprocessing phase that will allow us later to find $LCA(a,b)$ of two arbitrary nodes $a,b$ in a constant time and overall linear space, as in \cite{Bender2000}.

\myparagraph{Intermediate subtrees.}
We first show, how to retrieve pruned subtrees $l_G,G-l_G,L_G,G-L_G$ of a pruned subtree $G$, for the right side it will be symmetric.
Clearly, we already have $l_G$, because we represent the pruned subtree $G$ as a pair of its both main edges $l_G$ and $r_G$.
Similarly, $L_G$ is simply $\subtree(l_G^\uparrow)$ or $\subtree(l_G^\downarrow)$, depending on the position of $r_G$ with respect to $l_G$.

\FIGURE{h}{1}{lca}
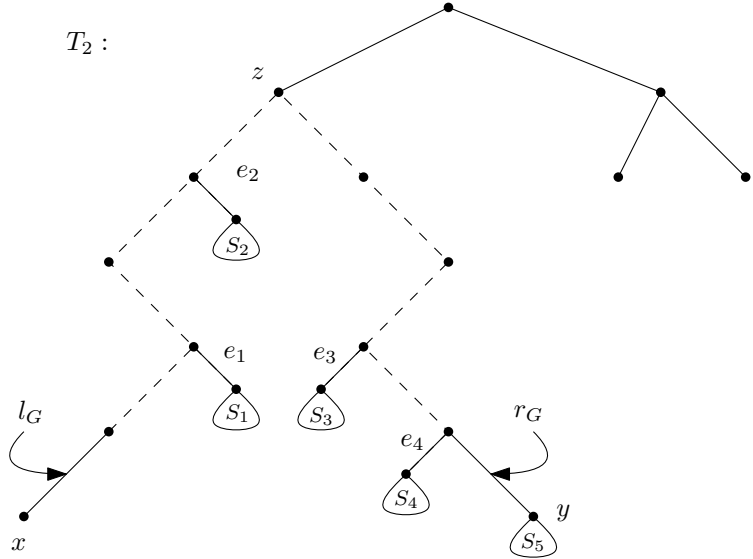
{A pruned subtree $G$ with main edges $l_G$ and $r_G$ has all the dashed edges contracted and to the tree belong edges $e_i$ and their subtrees $S_i$ for $i=1\ldots 5$.}

Let $x$ and $y$ be respectively the first and last nodes on the path from $l_G$ to $r_G$ and $z=LCA(x,y)$ be their lowest common ancestor.
See Figure~\ref{fig:lca}.
To retrieve $G'=G-L_G$ we only need to find its left main edge, because the right one does not change (provided that $l_G\ne r_G$).
There are two cases: either $l_{G'}$ is connected to the right of the path $(x\ldots z)$ or to the left of the path $(z\ldots y)$.
In the first case it is enough to remember for every node the first edge that is connected to the left and to the right to its path to the root of $\TT_2$ and then we can check if the edge is below the node $z$ or not.
Otherwise, let $t$ be the leftmost leaf in the right subtree of~$z$ (preprocessed, found by traversing down the tree going always left if possible, otherwise right).
In Figure~\ref{fig:lca} the node $t$ is inside subtree $S_3$.
Then $l_{G'}$ is the edge leading to the left child of $LCA(t,y)$.
As for the subtree $G-l_G$, if $L_G$ is empty, then $G-l_G=G-L_G$, otherwise we return the left main edge of $\subtree(l_G^\downarrow)$.

\myparagraph{Cost of contraction.}
Now we show how to retrieve the value of $\delta(\emptyset,G)$, the cost of contraction of all the edges from $G$.
Note that it is the sum of costs of contraction of all edges ``to the right'' of the path between $l_G$ and $r_G$ plus $c_{del}(l_G)+c_{del}(r_G)$.
For example, in Figure~\ref{fig:lca}, we need to contract $l_G$ and $r_G$ and all edges $e_i$ and their subtrees $S_i$ for $i=1\ldots 5$.
Observe that in order to contract all edges of $G$ we need to contract all edges to the right of the path $(z\ldots x)$ and to the left of $(z\ldots y)$.
Also notice that the path $(z\ldots x)$ is effectively the path $(root \ldots x)$ without its prefix $ (root \ldots z)$.
Thus we can use prefix sums and for every node store only the cost of contraction of all edges to the left or right to the path from the root of $\TT_2$ to the node.

Note that the above observations hold for all possible pruned subtrees of $\TT_2$, for instance also in the case for a subtree $G'$ such that $l_{G'}=r_G$ and $r_{G'}=l_G$ where $G$ is the subtree from Figure~\ref{fig:lca}.
To conclude, it is enough to remember a constant number of values in every node to be able to retrieve all intermediate pruned subtrees and compute the cost of contraction of all edges of a pruned subtree of~$\TT_2$ in a constant time.

\subsection{Computations in Limited Space}\label{se:limited_space}

In this subsection, we describe how to implement the $\computefrom$ procedure in $\Oh(nm)$ space.
Observe that each time we call $\computefrom$ subroutine, there is a set of pruned subtrees of one tree (either $\TT_1$ or $\TT_2$) and two pruned subtrees of the other: the initial and target.
For instance, when we call $\computefrom(\ALLL{\subtree(h_6^\uparrow)},\ALLL{\subtree(h_5^\uparrow)})$, then actually there are considered all pruned subtrees ``down'' $\TT_1$, $\subtree(h_5^\uparrow)$ is the initial tree and $\subtree(h_6^\uparrow)$ is target.
By a pruned subtree ``down'' $\TT_1$ we denote a pruned subtree obtained by a sequence of contractions of a main edge from the root, starting from $\TT_1$
\footnote{Recall that by $\TT_1$ in this context we again denote $\subtree(r_1)$, where $r_1$ is the dart corresponding to the initial rooting of $\TT_1$.}.
Clearly in this example there are considered $\Oh(m^2)$ pruned subtrees of $\TT_1$ and $\Oh(n)$ of $\TT_2$.
It is important that the target tree is obtained from the initial one by a sequence of uncontractions of a main edge always in the same direction.
Later on, we assume, that in this step we always uncontract the left main edge.

In the beginning we enumerate all pruned subtrees that are considered during this step (separately for $\TT_1$ and $\TT_2$) to be able to retrieve indices of subsequent trees in the dynamic program in constant time.
Now the difficulty lies in the fact that we cannot create the table of size $\Oh(nm^2)$ and we overcome it using an approach based on the one described by Demaine et al. \cite{DMRW}.
On a high level, we fix the right main edge of a pruned subtree $F$ of $\TT_1$ and consider all possible left main edges of $F$.
Then there are $\Oh(m)$ candidates for $l_F$ and $\Oh(n)$ candidates for pruned subtree $G$ of $\TT_2$.
The key insight is that while contracting the left main edge of a tree, its right main edge does not change unless it overlaps with the left one (which is the case when there is only one edge from the root).
Using this observation, we can store only the $\Oh(nm)$ values at any time.
However, we need to describe the details carefully.

We first describe in detail implementation of the $\computefrom$ subroutine for the case when the strategy considers only pruned subtrees ``down'' $\TT_1$, what is sufficient to implement all the algorithms for tree edit distance between rooted trees in $\Oh(nm)$ space.
Then we show how to handle also pruned subtrees ``up'' $\TT_2$, which is needed in our new algorithms for unrooted trees.

\myparagraph{Pruned subtrees ``down'' $\TT_1$.}
It might be easier to think, that now we describe how to implement $\computefrom(\ALLL{\subtree(h_5^\downarrow)},\ALLL{\subtree(h_6^\downarrow)}$ in $\Oh(nm)$ space.
In the beginning, we provide an equivalent definition of trees ``down'' $\TT_1$ that will be useful to implement the $\computefrom$ subroutine, see Figure~\ref{fig:tree_down} and Lemma~\ref{le:tree_down}.

\FIGURE{h}{1}{tree_down}
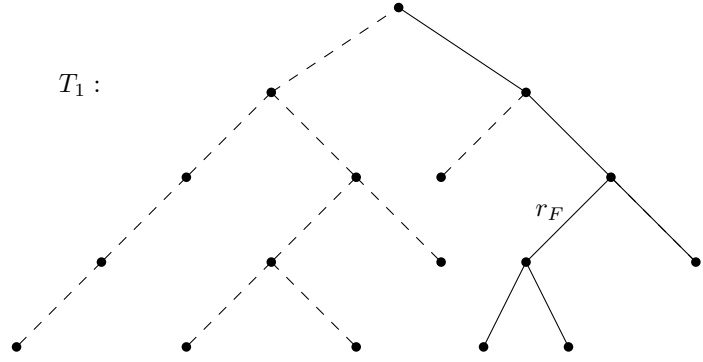
{Every pruned subtree $F$ ``down'' $\TT_1$ has its left main edge $l_F$ either equal to $r_F$ or to the left of the path from root to $r_F$.
All the candidates for $l_F$ are marked with the dashed edges.}

\begin{lemma}\label{le:tree_down}
 For every non-empty pruned subtree $F$ ``down'' $\TT_1$ holds that either $l_F=r_F$ or $l_F$ is strictly to the left of the path from the root to $r_F$.
\end{lemma}
\begin{proof}
 It is enough to show that for every pruned subtree $F$ with this property, $F-l_F$ and $F-r_F$ also have this property.
 This holds from the analysis of three cases: when $l_F=r_F$, $l_F\ne r_F$ and subtree of the contracted edge is non-empty or $l_F\ne r_F$ and subtree of the contracted edge is empty.
\end{proof}

Let $S_2$ be the set of all intermediate pruned subtrees of $\TT_2$ obtained by a sequence of uncontractions of the left main edge from the initial tree to target.
Now the algorithm considers candidates for the right main edge of the tree in $\TT_1$ in bottom-up order.
Then it computes edit distance between all pruned subtrees ``down'' $\TT_1$ with the specific right main edge and all trees from $S_2$.
It also needs to store explicitly values of $\delta(F,G)$ for trees $F$ of the form $\subtree(r^\downarrow)+\{r\}$ for $r\in\TT_1$ and $G\in S_2$, because the trees with both main edges overlapping need special attention.
See Algorithm~\ref{alg:memory_down}.

\begin{algorithm}[htb]
\begin{algorithmic}[1]
  \Function{$\computefrom$}{$\ALLL{\text{target}},\ALLL{\text{initial}}$}
  \State $S_2:=$ set of all intermediate pruned subtrees of $\TT_2$ between initial and target tree
  \State create arrays $C$ and $D$ of size $[m][n]$
  \State create array $RESULT$ of size $[m^2]$
  \ForEach{edge $r\in\TT_1$ in bottom-up order}
    \State $S_1:=\{F: \text{ ``down'' }\TT_1 \text{ and } r_F=r \}$
    \State $F'=\subtree(r^\downarrow)+\{r\}$
    \State compute $C[F',G]:=\delta(F',G)$ for all $G\in S_2$ \label{line:compute_f_prim}
    \State $RESULT[F']:=C[F',\text{target}]$
    \State create array $X$ of size $[m][n]$
    \State compute $X[F,G]:=\delta(F,G)$ for all $F\in S_1,G\in S_2$ \label{line:compute_rest}
    \ForEach{$F\in S_1$}
      \State $RESULT[F]:=X[F,\text{target}]$
    \EndFor
    \State $u:=$ the endpoint of $r$ that is closer to the root of $\TT_1$
    \ForEach{$G\in S_2$}
      \State $D[\TT_1^u,G]:=X[\TT_1^u,G]$
    \EndFor
  \EndFor
  \State \Return $RESULT$
 \EndFunction
\end{algorithmic}
\caption{Detailed description and implementation of $\computefrom$ procedure.}
\label{alg:memory_down}
\end{algorithm}

\noindent

Clearly, this subroutine runs in $\Oh(nm)$ space.
The arrays $C$ and $D$ are partially filled in every step of the main loop. $C$ stores edit distance between trees of $\TT_1$ with one edge from the root and $D$ is indexed by a tree $\TT_1^u$ where $u$ is the endpoint of $r$ that is closer to the root and a tree $G$ from $S_2$.
Finally, we need to show that all the required values during computations in lines \ref{line:compute_f_prim} and \ref{line:compute_rest} are available in the local arrays that we store.
In these lines, we process subtrees in the order of increasing sizes.
Recall that we assume that we always uncontract the left main edge.

First consider the step in line~\ref{line:compute_rest}.
While computing $\delta(F,G)$ for $F\in S_1,G\in S_2\setminus\{\text{initial}\}$ where $l_F\ne r_F$ we need to retrieve the value of $4$ subproblems and we show that each time it is available in one of the local arrays:

\begin{itemize}
 \item $\delta(F,G-l_G)=X[F,G-l_G]$, because $G-l_G\in S_2$,
 \item $\delta(F-l_{F},G)=X[F-l_{F},G]$, because $F-l_F\in S_1$,
 \item $\delta(F-L_{F},G-L_G)=\delta(F',G-L_G)=C[F',G-L_G]$, because $G-L_G\in S_2$,
 \item $\delta(L_{F},L_G)=\TAB[l_F^\downarrow,l_G^\downarrow]$ -- possibly $L_G\notin S_2$, so we need to use the value from $\TAB$ computed in an earlier stage.
\end{itemize}

\noindent
Similarly, to compute $\delta(F',G)$ for $G\in S_2\setminus\{\text{initial}\}$ in line~\ref{line:compute_f_prim} we have the following subproblems to consider:

\begin{itemize}
 \item $\delta(F',G-l_G)=C[F',G-l_G]$, because $G-l_G\in S_2$,
 \item $\delta(F'-l_{F'},G)=\delta(L_{F'},G)=D[L_{F'},G]$ which has already been computed, because we consider edges $r$ in bottom-up order,
 \item $\delta(F-L_{F'},G-L_G)=\delta(\emptyset,G-L_G)$ which we can retrieve in constant time after the preprocessing described in Section~\ref{se:preproc},
 \item $\delta(L_{F'},L_G)=\TAB[l_{F'}^\downarrow,l_G^\downarrow]$ -- as above.
\end{itemize}

In both variants, in the last case of $\delta(L_F,L_G)$ and $\delta(L_{F'},L_G)$ we used the values from the $\TAB$ table, which were computed by Demaine et al.'s algorithm in the very beginning.
However, we can also use the same implementation inside Demaine et al.'s algorithm, but then have to carefully analyze, that indeed the used values have already been computed and stored in the table.

Observe, that the very same implementation works even if there are two input tables, for instance in $\computefrom(\ALLL{\merged^R(A,M)}, \{\ALLL{\merged^R(A,B)};\ALLL{\subtree(h_{A-1}^\uparrow)}\})$. Note that in this case the set $S_2$ also contains the trees of the form $G-L_G$ and all the subsequent ones.
It needs to be slightly larger, because we need to ensure that for every $G\in S_2\setminus\{\text{initial}\}$ both $G-l_G$ and $G-L_G$ belong to $S_2$, which is the case as stated in Observation~\ref{le:nothing_missed}.
Similarly, note that it does not make any difference when we consider only pruned subtrees of $\TT_1$ of size bounded from above, i.e., smaller than $n/b$ (marked with $*$).

\myparagraph{Pruned subtrees ``down'' and ``up'' $\TT_2$.}
Now we need to slightly modify this approach, because in the $\insidegroup$ subroutine we consider also pruned subtrees ``up'' $\TT_2$: in line~\ref{line:avoid_hc} of Algorithm~\ref{alg:inside_a_group} we call $\computefrom(\delta(\TT_1^{u(i)},[l_G,r_G\in I]),\delta(\TT_1^{u(i+1)},[l_G,r_G\in I]))$.
In this case, we need to consider all possible pruned subtrees of $\TT_2$ defined by their two main edges from $I$.
We do it in two steps.
First is symmetric to the Algorithm~\ref{alg:memory_down} for all pruned subtrees ``down'' $\TT_2$, that is trees $G$ with $l_G$ to the left of the path from root of $\TT_2$ to $r_G$ (marked with dashed lines in Figure~\ref{fig:tree_down}) and the tree $\subtree(r_G^\downarrow)+\{r_G\}$ (with $l_G=r_G$).
The only difference is that now the roles of $\TT_1$ and $\TT_2$ are switched.

The second step is for all the remaining pruned subtrees of $\TT_2$, with the left main edge not to the left of the path from the root to $r_G$ (marked with solid lines in Figure~\ref{fig:tree_down}).
By the tree with $l_G=r_G$ we mean $\subtree(r_G^\uparrow)+\{r_G\}$.
It is done similarly, but now we need to consider edges $r$ in top-down order and store $D'[F,\subtree(e^\uparrow)]$, to be able to handle also the case of $G'=\subtree(r^\uparrow)+\{r\}$.
We also have to simultaneously fill the missing values of $\TAB[u,e^\uparrow]$ inside the procedure, because they have been already filled only for edges from the heavy path $H$, but not from the connected small subtrees.
To conclude, with these two steps we can implement the $\computefrom$ subroutine in $\Oh(nm)$ space.

We also need to elaborate more on the $\insidegroup$ procedure, in which there are considered pruned subtrees $G$ such that $l_G,r_G\in I$, but then the subsequent subtrees might have a main edge inside $I\cup D$.
However, we have already computed the values of these subproblems as mentioned in Observation~\ref{obs:edges_in_d}, so can retrieve them in a constant time.
Notice that from the computations in lines \ref{line:inside2_beg}-\ref{line:inside2_end} it is enough to store only the tables $\delta(\TUS,[r_G=h_{A-1}])$, $\delta(\TUS,[r_G=h_B])$, $\delta(\TUS,[l_G=h_{A-1}])$ and $\delta(\TUS,[l_G=h_B])$ respectively.
To sum up, also $\insidegroup$ and thus all the computations inside $\group$ procedure fit in the desired $\Oh(nm)$ space.

\subsection{Total Memory on Recursion Stack}

In the previous subsection we showed how to implement all the intermediate computations inside $\group,\insidegroup$ and $\computefrom$ in $\Oh(nm)$ space.
Clearly, these computations are disjoint, that is every time we run $\computefrom$ we can use the one and very same tables $C,D$ of size $\Oh(nm)$ and we only need to store separately inputs and outputs to the procedure.

Recall that during all the computations in $\TT_2$ we not always consider the whole tree $\TT_1$ and its all edges.
All the computations take into account the heavy path $\LH$ from $\TT_1$ and only the $\mm$ edges from the subtree $\TT_1^u$ of its top node $u$.
Then we have the following lemma about the size of tables passed to and from the functions:

\begin{lemma}
  For every call of function $\group,\insidegroup$ or $\computefrom$, size of input and returned tables is $\Oh(\mm^2)$. 
\end{lemma}
\begin{proof}
  The lemma clearly holds when we have a table of edit distance between a pruned subtree from $\TT_2$ and a set of pruned subtrees of $\TT_1^u$ because there are $\Oh(\mm^2)$ of them.
  The only situation in which we consider many pruned subtrees of $\TT_2$ is in $\computefrom(\delta(\TT_1^{u(i)},[l_G,r_G\in~I]), \delta(\TT_1^{u(i+1)},[l_G,r_G\in I]))$, but then there are also $\Oh(\mm^2)$ of them, as the set $I$ contains at most $\mm$ elements.
\end{proof}

Thus every recursive call pushes $\Oh(\mm^2)$ values on the stack.
From the analysis of Algorithm~\ref{alg:grouping_trick} we know that the depth of the recursion is $\Oh(1+\log \frac n\mm)$, so for the heavy path $\LH$ there are in total $\Oh(\mm^2\cdot (1+\log \frac n\mm))$ values stored on the recursion stack.
Using the inequality $\log x\le x$, we finally obtain that throughout the whole algorithm, there is $\Oh(n\mm)$ values on the stack.
Adding the auxiliary tables of the overall size $\Oh(nm)$ which are shared among $\computefrom$ calls, table~$\TAB$ and all the space used by Demaine et al.'s algorithm, we conclude that the whole algorithm computing edit distance between unrooted trees can be implemented in $\Oh(nm)$ space.

\begin{theorem}
 The algorithm computing edit distance between unrooted trees of sizes $n,m$ where $m \le n$ runs in $\Oh(nm^2(1+\log\frac nm))$ time and $\Oh(nm)$ space.
\end{theorem}

\section{Lower Bound}\label{se:lower_bound}

In this section, we restate known lower bounds for computing the edit distance between rooted trees (called rooted TED) and prove that they also hold for unrooted trees.
First, Demaine et al. \cite{DMRW} proved the following lower bound for decomposition algorithms:

\begin{theorem}[\cite{DMRW}]
For every decomposition algorithm for rooted TED and $n\ge m$, there exist trees $F$ and $G$ of sizes $\Theta(n)$ and $\Theta(m)$ such that the number of relevant subproblems is $\Omega(m^2n(1+\log\frac{n}{m}))$.
\end{theorem}
\noindent
which matches the complexity of the algorithm they provided.
Recently, Bringmann et al. \cite{TED_LowerBound} proved that a truly subcubic $\Oh(n^{3-\eps})$ algorithm for rooted TED is unlikely:

\begin{theorem}[\cite{TED_LowerBound}]
A truly subcubic algorithm for rooted TED on alphabet size $|\Sigma|=\Omega(n)$ implies a truly subcubic algorithm for APSP.
A truly subcubic algorithm for rooted TED on sufficiently large alphabet size $|\Sigma|=\Oh(1)$ implies an $\Oh(n^{k(1-\eps)})$ algorithm for Max-Weight $k$-Clique.
\end{theorem}


\subsection{Unrooted Case is Also Hard}

Now we show a reduction from rooted TED to the same problem for unrooted trees (unrooted TED).
It increases the number of nodes of a tree and size of the alphabet by a constant number, so the lower bounds from the rooted case will also apply for the unrooted case.

Given an instance $I=(\TT_1,\TT_2,\Sigma)$ of rooted TED we want to construct an instance $I'=(\TT_1',\TT_2',\Sigma')$ of unrooted TED such that, given an optimal solution of $I'$ it is possible to obtain an optimal solution of $I$.
Clearly, it is not enough to set $I'=I$, because it might be possible to change rooting of one of the trees (say $\TT_2'$) to obtain a smaller edit distance than between rooted $\TT_1$ and $\TT_2$.
That is actually the hardness in the problem of unrooted TED.

We need a gadget which ensures, that even if we allow all possible rootings, from every optimal rooting of $\TT_1'$ and $\TT_2'$ it is possible to obtain an optimal solution for $\TT_1$ and $\TT_2$.
It is enough to add one edge from the root as shown in Figure~\ref{fig:lower_bound} and appropriately set costs of contraction and relabeling of the fresh label $\$\notin \Sigma$ to force that edges with $\$$ are matched with their counterparts in the other tree.
More precisely, the costs are set as follows: $c_{del}(\$)=\infty$, $c_{match}(\$,\$)=0$ and $c_{match}(\$,\alpha)=\infty$ for $\alpha\ne \$$.

\FIGURE{h}{1}{lower_bound}
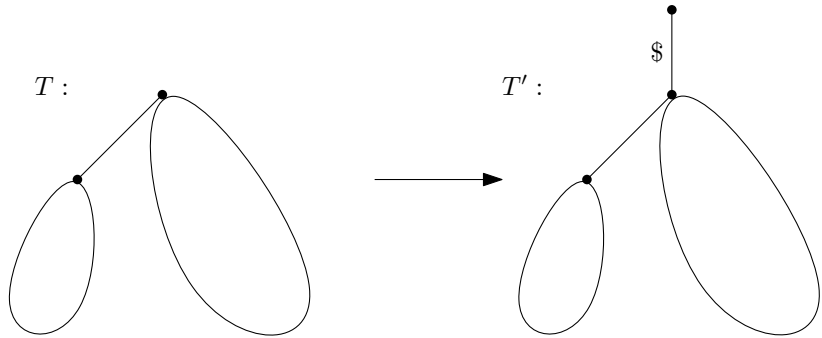
{A gadget changing an instance of rooted TED to unrooted TED.}

Clearly, in every optimal solution $OPT'$ of $I'$, the new edges with labels $\$$ are matched with each other.
Observe, that no matter how the trees are rooted in $OPT'$, we can rotate both trees simultaneously in such a way, that $\$$ is outgoing from the root as in Figure~\ref{fig:lower_bound}.
Informally, we can think of holding the trees by the edge with $\$$, with the original tree hanging down as in the initial rooting.

To conclude, our reduction adds only one new node to every tree and one new fresh label to the alphabet and allows retrieving an optimal solution of $I$ from an optimal solution of $I'$.
Thus, all lower bounds from the rooted TED hold also for unrooted TED.
Particularly, we proved that every decomposition algorithm for unrooted TED runs in $\Omega(m^2n(1+\log\frac{n}{m}))$ time which matches the complexity of our algorithm.
Finally, it is unlikely that there exists a truly subcubic $\Oh(n^{3-\eps})$ algorithm for unrooted TED.

\bibliography{biblio}

\appendix

\section{Demaine et al.'s \texorpdfstring{$\Oh(n^3)$}{O(n3)} Algorithm}\label{se:demaine}

In this section, we describe in detail the algorithm of Demaine et al. \cite{DMRW}, which is the fastest known algorithm for computing tree edit distance between rooted trees.
Our presentation is different from the original and adapted to the case of tree edit distance with labels on edges instead of nodes.

On a high level, their approach is to avoid the heavy child in currently larger of the two considered pruned subtrees, however once decided to avoid the heavy child in on of the trees, then they avoid the heavy child in this tree until it becomes empty.
Again, the algorithm uses the dynamic program from Lemma~\ref{le:dp} and to compute $\delta(F,G)$ needs to choose a direction, either left or right, for further computations.
It allows switching strategy to the other subtree only in the last case of the dynamic program in Lemma~$\ref{le:dp}$, in the recursive call of $\delta(R_F,R_G)$ or $\delta(L_F,L_G)$.

The strategy is to avoid the heavy child of either of the two trees $F$ or $G$ and the chosen tree will be denoted by $X$.
To choose the tree $X$ we need to check a slightly more complex condition: if either of $F$ or $G$ is empty, then $X$ is the non-empty of the two, otherwise if $|\TT_1^{\apex(F)}|>|\TT_2^{\apex(G)}|$ then $X$ is $F$, else $G$.
Recall that $\apex(F)$ is the top node of the heavy path containing the lowest common ancestor of all edges in $F$. See Algorithm~\ref{alg:Oren}.

\begin{algorithm}[htb]
\begin{algorithmic}[1]
  \Function{$\delta(F,G)$}{}
  \If{$F=G=\emptyset}$
  \Return $0$
  \EndIf
  \If{$G=\emptyset$ \textbf{ or } ($F\ne \emptyset$  \textbf{ and }  $|\TT_1^{\apex(F)}|>|\TT_2^{\apex(G)}|$)}
    \State $X:=F$
   \Else
     \State $X:=G$
  \EndIf
  \If{right child of the root of $X$ is not the heavy one}
    \State \Return $\min
    \begin{cases}
  \delta(F-r_F,G)+c_{del}(r_F) & \text{ if } F\ne\emptyset\\ 
  \delta(F,G-r_G)+c_{del}(r_G) & \text{ if } G\ne\emptyset\\
  \delta(R_F,R_G)+\delta(F-R_F,G-R_G)+c_{match}(r_F,r_G) & \text{ if } F,G\ne\emptyset\\   
 \end{cases}$
  \Else
    \State \Return $\min
    \begin{cases}
  \delta(F-l_F,G)+c_{del}(l_F) & \hspace{0.3cm} \text{ if } F\ne\emptyset\\ 
  \delta(F,G-l_G)+c_{del}(l_G) & \hspace{0.3cm} \text{ if } G\ne\emptyset\\
  \delta(L_F,L_G)+\delta(F-L_F,G-L_G)+c_{match}(l_F,l_G) & \hspace{0.3cm} \text{ if } F,G\ne\emptyset\\
 \end{cases}$
  \EndIf
  \EndFunction
\end{algorithmic}
\caption{Our presentation of Demaine et al.'s algorithm \cite{DMRW} for rooted TED.
}
\label{alg:Oren}
\end{algorithm}

Now we assume that both trees are of size $n$, $X=F$ and analyze the complexity of the strategy avoiding the heavy child in $X$.
Clearly, $|F|,|G|\le |\TT_1^{\apex(X)}|$.
Observe that $F$ was obtained from $\TT_1^{\apex(F)}$ by a sequence of successive contractions according to the strategy.
Next, $F-R_F$ is the tree $F$ with the right main edge contracted many times.
Thus, both $F-r_F$ and $F-R_F$ are also obtained from $\TT_1^{\apex(F)}$ by a sequence of successive contractions according to the strategy and $\apex(F-r_F)=\apex(F-R_F)=\apex(F)$.
Then, the only subsequent recursive call $\delta(F',G')$ in which $\apex(F')\ne \apex(F)$ is due to the recursive call $\delta(R_F,R_G)$ or $\delta(L_F,L_G)$.
It may also hold that $\apex(R_F) = \apex(F)$ when $F$ has only one child, and the strategy chooses the edge leading to the heavy child of the root.

We say, that a subproblem $\delta(F,G)$ is charged to the node $\apex(X)$.
Consider a node $v$ in $\TT_1$.
Suppose that $v$ is light because otherwise nothing is charged to it.
If $v$ is heavy, then there is no subproblem charged to $v$, so now suppose that $v$ is light.
From all the above observations we have, that among all subproblems $(F,G)$ charged to $v$, there are $\Oh(|T_1^v|)$ different pruned subtrees $F$ of $\TT_1$.
Next, all pruned subtrees $G$ of $\TT_2$ are not bigger than $|\TT_1^v|$ and each of them corresponds to an interval of Euler tour of length $n$.
Thus, there are $\Oh(n|T_1^v|)$ different pruned subtrees of $\TT_2$ among subproblems charged to node $v$ and in total there are $\Oh(n|T_1^v|^2)$ subproblems charged to a light node $v$ of $\TT_1$. 

Now by summing over all light nodes of $\TT_1$, we obtain that there are $\Oh(n\sum_{v: \text{ light node in }\TT_1}|\TT_1^v|^2)$ subproblems visited when $X=F$.
Because a symmetric argument holds for $X=G$, it remains to upper bound $t(n):=\sum_{v: \text{ light node in }\TT}|\TT^v|^2$ where $\TT$ is a tree of size $n$.
Denoting by $n_i$ the total size of the $i$-th subtree connected to the heavy path containing the root of $\TT$, we obtain the following bound: $t(n)\le n^2+\sum t(n_i)$.
It holds that $n_i\le (n-1)/2$ as the $i$-th subtree is rooted at a light node and $\sum_i n_i \le n-1$ as the subtrees connected to the heavy path are disjoint.
Now we prove by induction that $t(n)\le 2n^2$.

Using the inequality: $a^2+b^2\le (a-1)^2+(b+1)^2$ for $a\le b$ we can upper bound the sum $\sum_i n_i^2 $ with $2\cdot(n/2)^2$ iteratively choosing distinct indices $i,j$ such that $0<n_i\le n_j<(n-1)/2$, decreasing $n_i$ and increasing $n_j$ by $1$.
Combining it with the recurrence relation and the induction hypothesis we get:
$$t(n)=\sum_{v: \text{ light node in }\TT}|\TT^v|^2=n^2+\sum_i t(n_i) \le n^2+2\cdot \sum_i n_i^2 \le n^2+ 2\cdot2\cdot (n/2)^2 =2n^2$$
We conclude that there are $n\cdot t(n)=\Oh(n^3)$ subproblems visited when $X=F$ and similarly for $X=G$.
The algorithm can be also proved to run in $\Oh(nm^2(1+\log\frac nm))$ time for trees of unequal sizes $m\le n$, separately considering light nodes (apexes) $u$ such that $|T^u|\le m$ and $|T^u|>m$.

\end{document}